\documentclass[11pt]{llncs}
\usepackage{etex}

\usepackage{amssymb}
\usepackage{amsmath}
\usepackage{enumerate}
\usepackage{graphicx}
\usepackage{semantic}
\usepackage{mathtools}
\usepackage{proof}
\usepackage[colorlinks, citecolor=red]{hyperref}
\usepackage[all]{xy}

\newcommand{\act}[2][]{{\stackrel{#2}{\longrightarrow}}^{#1}}

\newcommand{\beq}{\simeq}

\usepackage{soul}
\setstcolor{red}
\newlength{\probwidth}
\renewcommand{\act}[2][]{\stackrel{#2}{\longrightarrow}^{#1}}

\usepackage{tikz}
\usepackage{fullpage}

\title{Dividing Line between Decidable PDA's and Undecidable Ones}
\author{Yuxi Fu \thanks{Email: \tt fu-yx@cs.sjtu.edu.cn} \and Qiang Yin\thanks{Email: \tt q.yin@sjtu.edu.cn} }
\institute{BASICS, Department of Computer Science, Shanghai Jiao Tong University}

\begin{document}

\maketitle

\begin{abstract}
S{\'e}nizergues has proved that language equivalence is decidable for disjoint $\epsilon$-deterministic PDA.
Stirling has showed that strong bisimilarity is decidable for PDA.
On the negative side Srba demonstrated that the weak bisimilarity is undecidable for normed PDA.
Later Jan\v{c}ar and Srba established the undecidability of the weak bisimilarity for disjoint $\epsilon$-pushing PDA and disjoint $\epsilon$-popping PDA.
These decidability and undecidability results are extended in the present paper.
The extension is accomplished by looking at the equivalence checking issue for the branching bisimilarity of several variants of PDA.
\end{abstract}

\section{Introduction}

\begin{quote}
``Is it recursively unsolvable to determine if $L_{1}=L_{2}$ for arbitrary deterministic languages $L_{1}$ and $L_{2}$?''

\hfill -- Ginsburg and Greibach, 1966
\end{quote}
The above question was raised in Ginsburg and Greibach's 1966 paper~\cite{GinsburgGreibach1966} titled Deterministic Context Free Languages.
The equality referred to in the above quotation is the language equivalence between context free languages.
It is well known that the context free languages are precisely those accepted by pushdown automata (PDA)~\cite{HopcroftUllman1979}.
A PDA extends a finite state automaton with a memory stack.
It accepts an input string whenever the memory stack is empty.
The operational semantics of a PDA is defined by a finite set of rules of the following form
\[pX\stackrel{a}{\longrightarrow}q\alpha\ \mathrm{or}\ pX\stackrel{\epsilon}{\longrightarrow}q\alpha.\]
The transition rule $pX\stackrel{a}{\longrightarrow}q\alpha$ reads ``If the PDA is in state $p$ with $X$ being on the top of the stack, then it can accept an input letter $a$, pop off $X$, place the string $\alpha$ of stack symbols onto the top of the stack, and turn into state $q$''.
The rule $pX\stackrel{\epsilon}{\longrightarrow}q\alpha$ describes a silent transition that has nothing to do with any input letter.
It was proved early on that language equivalence between pushdown automata is undecidable~\cite{HopcroftUllman1979}.
A natural question asks what restrictions one may impose on the PDA's so that language equivalence becomes decidable.
Ginsburg and Greibach studied deterministic context free languages.
These are the languages accepted by deterministic pushdown automata (DPDA)~\cite{GinsburgGreibach1966}.

A deterministic pushdown automaton enjoys disjointness and determinism properties.
The determinism property is the combination of $A$-determinism and $\epsilon$-determinism.
These conditions are defined as follows:
\begin{quote}
{\em Disjointness}.
For all state $p$ and all stack symbol $X$, if $pX$ can accept a letter then it cannot perform a silent transition, and conversely if $pX$ can do a silent transition then it cannot accept any letter.

$A$-{\em Determinism}.
If $pX\stackrel{a}{\longrightarrow}q\alpha$ and $pX\stackrel{a}{\longrightarrow}q'\alpha'$ then $q=q'$ and $\alpha=\alpha'$.

$\epsilon$-{\em Determinism}.
If $pX\stackrel{\epsilon}{\longrightarrow}q\alpha$ and $pX\stackrel{\epsilon}{\longrightarrow}q'\alpha'$ then $q=q'$ and $\alpha=\alpha'$.
\end{quote}
These are strong constraints from an algorithmic point of view.
It turns out however that the language problem is still difficult even for this simple class of PDA's.
One indication of the difficulty of the problem is that there is no size bound for equivalent DPDA configurations.
It is easy to design a DPDA such that two configurations $pY$ and $pX^{n}Y$ accept the same language for all $n$.

It was S\'{e}nizergues who proved after 30 years that the problem is decidable~\cite{Senizergues1997,Senizergues2001}.
His original proof is very long.
Simplified proofs were soon discovered by S\'{e}nizergues~\cite{Senizergues2002} himself and by Stirlng~\cite{Stirling2001-DPDA-dcidability}.
After the positive answer of S\'{e}nizergues, one wonders if the strong constraints (disjointness+$A$-determinism+$\epsilon$-determinism) can be relaxed.
The first such extension was given by S\'{e}nizergues himself~\cite{Senizergues1998}.
He showed that strong bisimilarity on the collapsed graphs of the disjoint $\epsilon$-deterministic pushdown automata is also decidable.
In the collapsed graphs all $\epsilon$-transitions are absorbed.
This result suggests that $A$-{\em non}determinism is harmless as far as decidability is concerned.
The silent transitions considered in~\cite{Senizergues1998} are $\epsilon$-popping.
A silent transition $pX\stackrel{\epsilon}{\longrightarrow}q\alpha$ is $\epsilon$-popping if $\alpha=\epsilon$.
In this paper we shall use a slightly more liberal definition of this terminology.
\begin{quote}
$\epsilon$-{\em Popping PDA}.
A PDA is $\epsilon$-popping if $|\alpha|\le1$ whenever $pX\stackrel{\epsilon}{\longrightarrow}q\alpha$.

$\epsilon$-{\em Pushing PDA}.
A PDA is $\epsilon$-pushing if $|\alpha|\ge1$ whenever $pX\stackrel{\epsilon}{\longrightarrow}q\alpha$.
\end{quote}
A disjoint $\epsilon$-deterministic PDA can be converted to an equivalent disjoint $\epsilon$-popping PDA in the following manner:
Without loss of generality we may assume that the disjoint $\epsilon$-deterministic PDA does not admit any infinite sequence of silent transitions.
Suppose $pX\stackrel{\epsilon}{\longrightarrow}\ldots\stackrel{\epsilon}{\longrightarrow}q\alpha$ and $q\alpha$ cannot do any silent transition.
If $\alpha=\epsilon$ then we can redefine the semantics of $pX$ by $pX\stackrel{\epsilon}{\longrightarrow}q\epsilon$; otherwise we can remove $pX$ in favour of $qZ$ with $Z$ being the first symbol of $\alpha$.
So under the disjointness condition $\epsilon$-popping condition is weaker than $\epsilon$-determinism.

A paradigm shift from a language viewpoint to a process algebraic viewpoint helps see the issue in a more productive way.
Groote and H\"{u}ttel~\cite{GrooteHuttel1994,Huttel1994} pointed out that as far as BPA and BPP are concerned the bisimulation equivalence \`{a} la Milner~\cite{Milner1989} and Park~\cite{Park1981} is more tractable than the language equivalence.
The best way to understand Senizergues' result is to recast it in terms of bisimilarity.
Disjointness and $\epsilon$-determinism imply that all silent transitions preserve equivalence.
It follows that the branching bisimilarity~\cite{vanGlabbeekWeijland1989-first-paper-bb} of the disjoint $\epsilon$-deterministic PDA's coincides with the strong bisimilarity on the collapsed graphs of these PDA's.
So what Senizergues has proved is that the branching bisimilarity on the disjoint $\epsilon$-deterministic PDA's is decidable.

The process algebraic approach allows one to use the apparatus from the process theory to study the equivalence checking problem for PDA.
Stirling's proof of the decidability of the strong bisimilarity for normed PDA (nPDA)~\cite{Stirling1996-nPDA-decidability,Stirling1998-nPDA-decidability} exploits the tableau method~\cite{HuttelStirling1991,Huttel1992}.
Later he extended the tableau approach to the study of the unnormed PDA~\cite{Stirling2000-PDA-decidability}.
Stirling also provided a simplified account of Senizergues' proof~\cite{Senizergues1998} using the process method~\cite{Stirling2001-DPDA-dcidability}.
The proof in~\cite{Stirling2001-DPDA-dcidability}, as well as the one in~\cite{Senizergues1998}, is interesting in that it turns the language equivalence of disjoint $\epsilon$-deterministic PDA to the strong bisimilarity of correlated models.
Another advantage of bisimulation equivalence is that it admits a nice game theoretical interpretation.
This has been exploited in the proofs of negative results using the technique of Defender's Forcing~\cite{JancarSrba2008}.
Srba proved that weak bisimilarity on nPDA's is undecidable~\cite{Srba2002d}.
Jancar and Srba improved this result by showing that the weak bisimilarity on the disjoint nPDA's with only $\epsilon$-popping transitions, respectively $\epsilon$-pushing transitions, is already undecidable~\cite{JancarSrba2008}.
In fact they proved that the problems are $\Pi_{1}^{0}$-complete.
Recently Yin, Fu, He, Huang and Tao have proved that the branching bisimilarity for all the models above either the normed BPA or the normed BPP in the hierarchy of process rewriting system~\cite{Mayr2000PRS} are undecidbale~\cite{YinFuHeHuangTao2014}.
This general result implies that the branching bisimilarity on nPDA is undecidable.
The idea of Defender's Forcing can also be used to prove complexity bound.
An example is Benedikt, Moller, Kiefer and Murawski's proof that the strong bisimilarity on PDA is non-elementary~\cite{BenediktMollerKieferMurawski2013}.
A summary of the decidability/undecidability results mentioned above is given in Fig.~\ref{Decidability-of-PDA} and Fig.~\ref{More-on-Decidability-of-PDA}.

\begin{figure*}[t]
\begin{center}
\begin{tabular}{|c|c|c|} \hline
& {\bf PDA} & {\bf nPDA} \\ \hline\hline
$\sim$ & \begin{tabular}{c} Decidable~\cite{Senizergues1998,Stirling2000-PDA-decidability} \\
 Non-Elementary~\cite{BenediktMollerKieferMurawski2013} \end{tabular} & \begin{tabular}{c} Decidable~\cite{Stirling1996-nPDA-decidability,Stirling1998-nPDA-decidability} \\ Non-Elementary~\cite{BenediktMollerKieferMurawski2013} \end{tabular} \\ \hline
$\simeq$ & Undecidable~\cite{YinFuHeHuangTao2014} & Undecidable~\cite{YinFuHeHuangTao2014} \\ \hline
$\;\approx\;$ & \begin{tabular}{c} $\Sigma^{1}_{1}$-Complete~\cite{JancarSrba2008} \\ Undecidable~\cite{Srba2002d} \end{tabular} & \begin{tabular}{c} $\Sigma^{1}_{1}$-Complete~\cite{JancarSrba2008} \\ Undecidable~\cite{Srba2002d} \end{tabular} \\ \hline
\end{tabular}
\end{center}
\caption{Decidability of PDA \label{Decidability-of-PDA}}
\end{figure*}

\begin{figure*}[t]
\begin{center}
\begin{tabular}{|c|c||c|c|} \hline
& $\epsilon$-Popping {\bf nPDA}/{\bf PDA} & $\epsilon$-Pushing {\bf nPDA} & $\epsilon$-Pushing {\bf PDA} \\ \hline\hline
$\simeq$ & ? & ? & ? \\ \hline
$\;\approx\;$ & $\Pi_{1}^{0}$-Complete~\cite{JancarSrba2008} & $\Pi_{1}^{0}$-Complete~\cite{JancarSrba2008} & $\Sigma^{1}_{1}$-Complete~\cite{JancarSrba2008} \\ \hline\hline
\end{tabular}
\end{center}
\caption{More on Decidability of PDA \label{More-on-Decidability-of-PDA}}
\end{figure*}

The decidability of the strong bisimilarity and the undecidability of the weak bisimilarity still leaves a number of questions unanswered.
A conservative extension of the language equivalence for DPDA is neither the strong bisimilarity nor the weak bisimilarity.
It is not the former because language equivalence ignores silent transitions.
It is not the latter since the whole point of introducing the disjointness and $\epsilon$-determinism conditions is to force all silent transitions to preserve equivalence.
To investigate the possibility of extending the decidability result of DPDA, one should really start with the branching bisimilarity.
This is what we are going to do in this paper.
Since Senizergues' result can be stated as saying that the branching bisimilarity on the disjoint $\epsilon$-deterministic PDA is decidable, we will look at the situations in which either the disjointness condition is dropped and/or the $\epsilon$-determinism condition is weakened/removed.
It turns out that both the decidability result and the undecidability about PDA can be strengthened.

The contributions of this paper are summarized as follows.
\begin{enumerate}
\item Technically we will provide answers to some of the open problems raised in literature.
The main results are the following.
\begin{itemize}
\item The branching bisimilarity on the $\epsilon$-popping PDA is decidable.
\item The branching bisimilarity on the $\epsilon$-pushing nPDA is decidable.
\item The branching bisimilarity on the $\epsilon$-pushing PDA is $\Sigma_1^1$-complete.
\end{itemize}
\item At the model theoretical level we propose a model that strictly extends the classical PDA model.
The new model gets rid of the notion of stack in favour of a structural definition of processes.
The structural definition helps simply the proofs of our results significantly.
\end{enumerate}

The rest of the paper is organised as follows.
Section~\ref{sec-Syntax-and-Semantics-of-PDA} introduces an extended PDA model.
Section~\ref{sec-Branching-Bisimilarity} reviews the basic properties of the branching bisimilarity.
Section~\ref{sec-Finite-Branching-Property} confirms that the finite branching property hold for both the $\epsilon$-pushing nPDA and the $\epsilon$-popping PDA.
Section~\ref{sec-Decidability of PDAminus} establishes the decidability of the $\epsilon$-popping PDA.
Section~\ref{sec-Decidability-of-nPDAplus} points out that the proofs given in Section~\ref{sec-Decidability of PDAminus} can be repeated for the $\epsilon$-pushing nPDA.
Section~\ref{sec-High-Undecidability} applies the Defender's Forcing technique to show that $\epsilon$-nondeterminism is highly undecidable.
Section~\ref{sec-Conclusion} concludes with remark on future work.

\section{Syntax and Semantics of PDA}\label{sec-Syntax-and-Semantics-of-PDA}

A {\em pushdown automaton} (or simply PDA) $\Gamma=(\mathcal{Q},\mathcal{V},\mathcal{L},\mathcal{R})$ consists of
\begin{itemize}
\item
a finite set of states $\mathcal{Q}=\{p_{1},\ldots,p_{\mathfrak{q}}\}$ ranged over by $o,p,q,r,s,t$,
\item
a finite set of symbols $\mathcal{V}=\{X_{1},\ldots,X_{\mathfrak{n}}\}$ ranged over by $X,Y,Z$,
\item
a finite set of letters $\mathcal{L}=\{a_{1},\ldots,a_{\mathfrak{s}}\}$ ranged over by $a,b,c,d$, and
\item
a finite set of transition rules $\mathcal{R}$.
\end{itemize}
If we think of a PDA as a process we may interpret a letter in $\mathcal{L}$ as an action label (or simply action).
The set $\mathcal{L}^{*}$ of words is ranged over by $u,v,w$.
Following the convention in language theory a silent transition will be denoted by $\epsilon$.
The set $\mathcal{A}=\mathcal{L}\cup\{\epsilon\}$ of actions is ranged over by $\ell$.
The set $\mathcal{A}^{*}$ of action sequence is ranged over by $\ell^{*}$.
The set $\mathcal{V}^{*}$ of strings of symbols is ranged over by small Greek letters.
By overloading notation the empty string is also denoted by $\epsilon$.
We identify both $\epsilon\alpha$ and $\alpha\epsilon$ to $\alpha$ syntactically.
The length of $\alpha$ is denoted by $|\alpha|$.

A {\em pushdown process}, or {\em PDA process}, is an interactive object that has a syntactical tree structure.
Pushdown processes are defined in terms of {\em constants}.
A constant is a finite list of PDA processes.
To help study the algebraic properties of PDA processes, it is convenient to introduce a special class of constants called {\em recursive constants}.
A recursive constant is defined by a recursive equation.
Formally the set of the PDA processes and the set of the constants definable in a PDA $\Gamma$ are generated from the following BNF:
\begin{eqnarray*}
P &\ :=\ & {\bf 0} \; \mid\; l \; \mid\; pXC_{[n]} \; \mid\; l\cdot C_{[n]}, \\
C_{[n]} &\ :=\ & (P_{1},\ldots,P_{n}) \; \mid\; V_{[n]}.
\end{eqnarray*}
In the above definition, $l$ ranges over the set $\mathbb{N}$ of {\em positive} integer and the notation $[n]$ stands for the set $\{1,2,\ldots,n\}$.
An $n$-{\em ary} constant $C_{[n]}$ is either an $n$-tuple of processes or an $n$-{\em ary} recursive constant.
An alternative notation for $(P_{1},\ldots,P_{n})$ is $\left(P_{i}\right)_{i\in[n]}$.
A unary constant $(P)$ is identified with the process $P$ syntactically.
We sometimes omit the subscript in $C_{[n]}$ when no confusion may arise.
A process is either the {\em nil} process ${\bf 0}$, or a {\em selection} process $l$, or a {\em sequential} process $pXC_{[n]}$, or a {\em composition} process $l\cdot C_{[n]}$.
In the composition process $l\cdot C_{[n]}$ the role of the process $l$ is to select a continuation from $C_{[n]}$.
We impose the following equality on composition processes:
\begin{eqnarray*}
l\cdot (P_{1},\ldots,P_{n}) &=& \left\{\begin{array}{ll}
P_{i}, & \ \mathrm{if}\ l\in[n], \\
l, & \ \mathrm{otherwise};
\end{array}\right. \\
l\cdot V_{[n]} &=& l,\ \ \mathrm{if}\ l>n.
\end{eqnarray*}
Throughout this paper the equality symbol ``$=$'' stands for grammar equality.
So $l\cdot (P_{1},\ldots,P_{n})$ is syntactically identified to $P_{i}$ if $l\in[n]$ and to $l$ if $l>n$.
Similarly $l\cdot V_{[n]}$ is syntactically the same as $l$ if $l>n$.
We sometimes write $C_{[n]}(i)$ for $i\cdot C_{[n]}$ if $i\in[n]$.

It will be very helpful to generalize the composition operation between a selection process and a constant to one between a process/constant and a constant.
The operation is defined by the following induction.
\begin{eqnarray*}
{\bf 0}\cdot C' &=& {\bf 0}, \\
(pXC)\cdot C' &=& pX(C\cdot C'), \\
V_{[n]}(i)\cdot C &=& V_{[n]}(i), \\
(P_{1},\ldots,P_{n})\cdot C &=& (P_{1}\cdot C,\ldots,P_{n}\cdot C), \\
V_{[n]}\cdot C &=& V_{[n]}.
\end{eqnarray*}
Using simple induction it is easy to show that the generalized composition operator is associative.
In the sequel we shall omit the composition operator.
So the concatenation $PC$, respectively $CC'$, is the composition $P\cdot C$, respectively the composition $C\cdot C'$.

Before explaining recursive constants, we need to introduce the function $ln(\_)$ defined as follows.
\begin{eqnarray*}
ln({\bf 0}) &=& \emptyset, \\
ln(l) &=& \{l\}, \\
ln(pX(P_{1},\ldots,P_{n})) &=& ln(P_{1})\cup\ldots\cup ln(P_{n}), \\
ln(V_{[n]}(i)) &=& \emptyset, \\
ln((P_{1},\ldots,P_{n})) &=& ln(P_{1})\cup\ldots\cup ln(P_{n}), \\
ln(V_{[n]}) &=& \emptyset.
\end{eqnarray*}
The function returns the set of dangling selection processes appearing in a process/constant.

A process or constant is {\em simple} if it contains no occurrences of recursive constant.
A recursive constant $V_{[n]}$ is defined by an equation of the form
\begin{equation}\label{2014-02-15}
V_{[n]} = (P_{1},\ldots,P_{n})\cdot V_{[n]}
\end{equation}
such that the following statements are valid:
\begin{itemize}
\item
$P_{i}$ is simple for each $i\in[n]$.
\item
$ln(P_{i})\subseteq[n]$ for each $i\in[n]$.
\item
If $P_i$ is a selection process, it must be $i$.
\end{itemize}
The grammar equality  $V_{[n]}(i)=P_{i}V_{[n]}$ is derivable by our definition of the extended composition operator.
If $P_i=i$ then $V_{[n]}(i)=P_{i}V_{[n]}$ is the trivial identity.
In this case we impose the grammar equality $V_{[n]}(i)={\bf 0}$.
In the terminology of PDA the constant $V_{[n]}$ is a stack that has recursive behaviour.
The simplest $n$-ary recursive constant $I_{[n]}$ is defined by $I_{[n]}=(1,\ldots,n)I_{[n]}$, which is the same as the $n$-ary constant $({\bf 0},\ldots,{\bf 0})$.
The terminologies `simple constant' and `recursive constant' are introduced in Stirling's work~\cite{Stirling1996-nPDA-decidability}.
The constants introduced in this paper are more general and have better composition property.
In the sequel we will write $L,M,N,O,P,Q$ for processes, $A_{[n]},B_{[n]},C_{[n]},D_{[n]}$ for constants, $U_{[n]}$ for simple constant, and $V_{[n]}$ for recursive constant.

We now introduce the auxiliary notation $p\alpha C_{[n]}$, which will make evident the relationship between the standard PDA and our extended PDA.
The process abbreviated by $p\alpha C_{[n]}$ is inductively defined as follows:
(i) $p_{i}\epsilon = i$, and (ii) $pX\beta C = pX(p_{1}\beta C,\ldots,p_{\mathfrak{q}}\beta C)$.
Now let $p\alpha = p\alpha(p_{1}\epsilon,\ldots,p_{\mathfrak{q}}\epsilon)$.
In this way a standard PDA process $p\alpha$ can be seen as an abbreviation of a simple pushdown process written in our notation.

Finally we are in a position to define the operational semantics of PDA.
Every transition rule in $\mathcal{R}$ is of the form $pX\stackrel{\ell}{\longrightarrow}q\alpha$.
The transition semantics is defined by the following two rules:
\begin{eqnarray*}
\inference{pX\stackrel{\ell}{\longrightarrow}q\alpha \in \mathcal{R}}{pX\stackrel{\ell}{\longrightarrow}q\alpha}\ \ \  &
\inference{pX\stackrel{\ell}{\longrightarrow}q\alpha}{pXC_{[n]}\stackrel{\ell}{\longrightarrow}q\alpha C_{[n]}}
\end{eqnarray*}
We shall use the standard notation $\stackrel{\ell^{*}}{\longrightarrow}$ and $\Longrightarrow$ and $\stackrel{\ell^{*}}{\Longrightarrow}$.
A process of the form $l$ is an accepting process, which cannot perform any action.
A process $P$ {\em accepts} a word $w$ if $P\stackrel{w}{\Longrightarrow}l$ for some $l$.
A process $P$ is {\em normed}, or $P$ is an nPDA process, if $P\stackrel{\ell^{*}}{\longrightarrow}l$ for some $\ell^{*},l$.
A PDA $\Gamma=(\mathcal{Q},\mathcal{V},\mathcal{L},\mathcal{R})$ is normed, or $\Gamma$ is an nPDA, if $pX$ is normed for all $p\in\mathcal{Q}$ and all $X\in\mathcal{V}$.
The notation $\mathrm{PDA}^{\epsilon+}$ will refer to the variant of PDA with $\epsilon$-pushing transitions, and $\mathrm{nPDA}^{\epsilon-}$ to the variant of nDPA with $\epsilon$-popping transitions.

At a more intuitive level a process can be identified to a possibly infinite $\mathfrak{q}$-branching labeled tree with an internal node labeled by $pX$ for some $p\in\mathcal{Q},X\in\mathcal{V}$ and a leaf labeled by either a positive number or ${\bf 0}$.
Regarded as a tree the composition $C_{[n]}C_{[m]}$ is obtained by replacing every leaf of $C_{[n]}$ with label $l\le m$ by the tree of $C_{[m]}(l)$ and leaving every leaf of $C_{[n]}$ with label $l>m$ unchanged.
Using the tree interpretation a constant $C_{[n]}$ can be {\em decomposed} into a {\em simple} constant $C_{[n]}'$ followed by another constant $C_{[m]}''$, i.e. $C_{[n]}=C_{[n]}'C_{[m]}''$.
Such a decomposition is of course not unique.
Normally we should specify the arity of $C_{[m]}'$.
We will not formally define the notion of decomposition.
For the purpose of this paper it is sufficient to see an example.
Consider a simple constant $C   = (p_1X_1(l_1, p_2X_2(l_2,l_3)), q_1Y_1(q_2Y_2(l_4,l_5), {\bf 0}))$.
Diagrammatically $C$ can be depicted as the following two trees.
\begin{center}
  \resizebox{.7\textwidth}{!}{
    \begin{tikzpicture}
      \node (p6) at (3,0) {$l_2$};
      \node (p7) at (5,0) {$l_3$};

      \node (q4) at (7,0) {$l_4$};
      \node (q5) at (9,0) {$l_5$};

      \node (p2) at (1,2) {$l_1$};
      \node (p3) at (4,2) {$p_2X_2$};
      \node (q2) at (8,2) {$q_2Y_2$};
      \node (q3) at (11,2) {${\bf 0}$};

      \node (p1) at (2.5,4) {$p_1X_1$};
      \node (q1) at (9.5,4) {$q_1Y_1$};
      \draw (p1) -- (p2);
      \draw (p1) -- (p3);
      \draw (p3) -- (p6);
      \draw (p3) -- (p7);

      \draw (q1) -- (q2);
      \draw (q1) -- (q3);
      \draw (q2) -- (q4);
      \draw (q2) -- (q5);

      \node (d1) at (12.5, 2) {$(1)$};
      \node (d2) at (10.5, 0) {$(2)$};

      \draw [dashed] (-.4,2) -- (p2) --(p3)-- (q2) -- (q3) -- (d1);
      \draw [dashed] (1.6,0 ) -- (p6) -- (p7) -- (q4) -- (q5) -- (d2);
    \end{tikzpicture}}
\end{center}
Two decompositions of $C$ are $C_{11}C_{12}$ and $C_{21}C_{22}$, where $C_{11}$, $C_{12}$, $C_{21}$, $C_{22}$ are defined as follows.
\begin{itemize}
\item [(1)] $C_{11}=(p_1X_1(1,2), q_1Y_1(3,4))$ and $C_{12}=(l_1, p_2X_2(l_2,l_3), q_2Y_2(l_4,l_5),{\bf 0})$; \vspace{.5em}
\item [(2)] $C_{21}=(p_1X_1(l_1,p_2X_2(1,2)), q_1Y_1(q_2Y_2(3,4),{\bf 0}))$ and $C_{22}=(l_2,l_3,l_4,l_5)$.
\end{itemize}
The decompositions are indicated by the dashed lines in the above diagram.
Given a process $P$ and $k>0$, its {\em simple prefix $P{\upharpoonright}_{k}$ up to level $k>0$}, which must be a simple process, is defined as follows:
\begin{eqnarray*}
l{\upharpoonright}_{k} &=& l, \\
{\bf 0}{\upharpoonright}_{k} &=& {\bf 0}, \\
(pX(P_{1},\ldots,P_{n})){\upharpoonright}_{k} &=& \left\{ \begin{array}{ll} pX(1,\ldots,n), & \mathrm{if}\ k=1, \\ pX(P_{1}{\upharpoonright}_{k-1},\ldots,P_{n}{\upharpoonright}_{k-1}), & \mathrm{if}\ k>1. \end{array} \right.
\end{eqnarray*}
Using the above example we see that $C{\upharpoonright}_{1}=C_{11}$ and $C{\upharpoonright}_{2}=C_{21}$.
The operation $(\_){\upharpoonright}_{k}$ can be applied to a constant with the obvious effect.

The {\em height} of a simple process $pX(P_{1},\ldots,P_{n})$ is defined as follows:
\begin{eqnarray*}
|{\bf 0}| &=& 0,\\
|l| &=& 0,\\
|pX(P_{1},\ldots,P_{n})| &=& 1+\max_{1\le i\le n}\{|P_{i}|\}.
\end{eqnarray*}
The height of a simple constant $(P_{1},\ldots,P_{n})$ is the maximum of the height of all $P_{1},\ldots,P_{n}$.
If we replace the $\max$ operation in the above definition by the $\min$ operation, we get the {\em minimal height} $\min|Q|$ of a process $Q$.
Given a process $P$ and a number $k>0$, we can decompose $P$ into $P{\upharpoonright}_{k}\cdot C$ for some constant $C$ such that $P{\upharpoonright}_{k}$ is a simple process with $|P{\upharpoonright}_{k}| =k$.

\section{Branching Bisimilarity}\label{sec-Branching-Bisimilarity}

The definition of branching bisimilarity is due to van Glebbeek and Weijland~\cite{vanGlabbeekWeijland1996}.
Our definition of branching bisimilarity for PDA is similar to Stirling's definition given in~\cite{Stirling1996-nPDA-decidability}.
\begin{definition}\label{2014-03-30}
A symmetric relation $\mathcal{R}$ on PDA processes is a {\em branching bisimulation} if the following statements are valid:
\begin{enumerate}
\item If $P\mathcal{R}Q\stackrel{a}{\longrightarrow}Q'$ then $P\Longrightarrow P''\stackrel{a}{\longrightarrow}P'\mathcal{R}Q'$ and $P''\mathcal{R}Q$ for some $P',P''$.
\item If $P\mathcal{R}Q\stackrel{\epsilon}{\longrightarrow}Q'$ then either $P\mathcal{R}Q'$ or $P\Longrightarrow P''\stackrel{\epsilon}{\longrightarrow}P'\mathcal{R}Q'$ for some $Q',Q''$ such that $P''\mathcal{R}Q$.
\item $P=j$ if and only if $Q=j$.
\end{enumerate}
The {\em branching bisimilarity} $\simeq$ is the largest branching bisimulation.
\end{definition}
The condition 3 in Definition~\ref{2014-03-30} guarantees that $\simeq$ is closed under composition.
We write $\simeq_{\mathrm{nPDA}^{\epsilon+}}$ for example for the branching bisimilarity on $\mathrm{nPDA}^{\epsilon+}$ processes.

A technical lemma that plays an important role in the study of branching bisimilarity is the Computation Lemma~\cite{vanGlabbeekWeijland1996}.

\begin{lemma}\label{computation-lemma}
If $P_{0}\stackrel{\epsilon}{\longrightarrow}P_{1}\stackrel{\epsilon}{\longrightarrow}\ldots \stackrel{\epsilon}{\longrightarrow}P_{k}\simeq P_{0}$, then $P_{0}\simeq P_{1}\simeq\ldots\simeq P_{k}$.
\end{lemma}

A silent transition $P\stackrel{\epsilon}{\longrightarrow}P'$ is {\em state-preserving}, notation $P\rightarrow P'$, if $P\simeq P'$.
It is a {\em change-of-state}, notation $P\stackrel{\iota}{\longrightarrow}P'$, if $P\not\simeq P'$.
We write $\rightarrow^{*}$ for the reflexive and transitive closure of $\rightarrow$.
The notation $P\nrightarrow$ stands for the fact that $P\not\simeq P'$ for all $P'$ such that $P\stackrel{\epsilon}{\longrightarrow}P'$.
Let $\jmath$ range over $\mathcal{L}\cup\{\iota\}$.
We will find it necessary to use the notation $\stackrel{\jmath}{\longrightarrow}$.
The transition $P\stackrel{\jmath}{\longrightarrow}P'$ refers to either $P\stackrel{a}{\longrightarrow}P'$ for some $a\in\mathcal{L}$ or $P\stackrel{\iota}{\longrightarrow}P'$.
Lemma~\ref{computation-lemma} implies that if $P_{0}\stackrel{\jmath}{\longrightarrow}P_{1}$ is bisimulated by $Q_{0}\stackrel{\epsilon}{\longrightarrow}Q_{1}\stackrel{\epsilon}{\longrightarrow}\ldots\stackrel{\epsilon}{\longrightarrow}Q_{k}\stackrel{\jmath}{\longrightarrow}Q_{k+1}$,
then $Q_{0}\rightarrow Q_{1}\rightarrow\ldots\rightarrow Q_{k}$.
This property of the branching bisimilarity will be used extensively.

Given a PDA process $P$, the {\em norm} of $P$, denoted by $\|P\|$, is a function from $\mathbb{N}$ to $\mathbb{N}\cup\{\bot\}$ such that the following holds:
\begin{itemize}
\item $\|P\|(h)=\bot$ if and only if there does not exist any $\ell^{*}$ such that $P\stackrel{\ell^{*}}{\longrightarrow}h$.
\item $\|P\|(h)$ is the least number $i$ such that $\exists \jmath_{1}\ldots\jmath_{i}.\;P\rightarrow^{*}\stackrel{\jmath_{1}}{\longrightarrow}\ldots\rightarrow^{*}\stackrel{\jmath_{i}}{\longrightarrow}\rightarrow^{*}h$.
\end{itemize}
It is easy to see that $\mathrm{def}\,\|P\|=\{h\mid \|P\|(h)\ne\bot\}$ is finite.
A process $P$ is {\em normed} if $\mathrm{def}\,\|P\|\ne\emptyset$.
It is {\em unnormed} otherwise.
It is clear that every process of the form $V_{[n]}(i)$ is unnormed, where $i\in[n]$.
For normed process $P$ we write $\min\|P\|$ for $\min\{\|P\|(h)\mid h\in\mathrm{def}\,\|P\|\}$.

We shall use the following convention in the rest of the paper.
\begin{itemize}
\item $\mathfrak{r}=\max\left\{|\eta| \left|\; pX\stackrel{\ell}{\longrightarrow}q\eta\in\Delta\ \mathrm{for}\ \mathrm{some}\ p,q,X\right\}\right.$;
\item $\mathfrak{m}=\max\left\{\min_{i}\{i\mid \exists h.\exists\ell_{1}\ldots\ell_{i}. pX\stackrel{\ell_{1}}{\longrightarrow}\ldots\stackrel{\ell_{i}}{\longrightarrow}h\} \mid pX\ \mathrm{is}\ \mathrm{normed}\right\}$.
\end{itemize}
Both $\mathfrak{r}$ and $\mathfrak{m}$ can be effectively calculated.
By definition $\|pX\|(i)\le\mathfrak{m}$ for all $p,X$ and all $i\in\mathrm{def}\,\|pX\|$.

If $qX\stackrel{\epsilon}{\longrightarrow}q_{1}X_{1}\stackrel{\epsilon}{\longrightarrow}\ldots\stackrel{\epsilon}{\longrightarrow}q_{k}X_{k}\stackrel{\epsilon}{\longrightarrow}pX$, then $qXC\simeq q_{1}X_{1}C\simeq\ldots\simeq q_{k}X_{k}C$ for all $C$.
In this case we get an equivalent PDA by replacing $X_{1},\ldots,X_{k}$ by $X$.
From now on we assume that such circularity does not occur in our PDA.
This means that in a PDA with $\mathfrak{n}$ variables and $\mathfrak{q}$ states the length of a silent transition sequence of the form $qX\stackrel{\epsilon}{\longrightarrow}q_{1}X_{1}\stackrel{\epsilon}{\longrightarrow}\ldots\stackrel{\epsilon}{\longrightarrow}q_{k}X_{k}$ is less than $\mathfrak{n}\mathfrak{q}$.

In some proofs to be given later we need to use the game theoretical interpretation of bisimulation.
A {\em bisimulation game}~\cite{Stirling1998,JancarSrba2008} for a pair of processes $(P_0,P_1)$, called a {\em configuration}, is played between Attacker and Defender in an alternating fashion.
It is played according to the following rules:
\begin{enumerate}
\item Beginning with the configuration $(P,Q)$, Attacker picks up some $P_i$, where $i\in\{0,1\}$, to start with.
It then chooses some transition $P_i\stackrel{\ell}{\longrightarrow}P_i'$.
\item Defender must respond in the following manner:
\begin{enumerate}
\item Do nothing. This option is available if $\ell=\epsilon$.
\item Choose a transition sequence $P_{1-i}\stackrel{\epsilon}{\longrightarrow}P_{1-i}^{1}\stackrel{\epsilon}{\longrightarrow}\ldots\stackrel{\epsilon}{\longrightarrow}P_{1-i}^{k-1}\stackrel{\ell}{\longrightarrow}P_{1-i}^{k}$.
\end{enumerate}
\item If case 2(a) happens the new configuration is $(P_i,P_{1-i})$.
If case 2(b) happens Attacker chooses one of $\{(P_i,P_{1-i}^{1}),\ldots,(P_i,P_{1-i}^{k-1}),(P_i',P_{1-i}^{k})\}$ as the new configuration.
\item The game continues with the new configuration.
\end{enumerate}
Attacker wins a bisimulation game if Defender gets stuck in the game.
Defender wins a bisimulation game if Attacker does not win the game.
Attacker/Defender has a winning strategy if it can win no matter how its opponent plays.
The effectiveness of the bisimulation game is enforced by the following lemma.

\begin{lemma}
$P\simeq Q$ if and only if Defender has a winning strategy for the bisimulation game starting with the configuration $(P,Q)$.
\end{lemma}

This lemma is the basis for game theoretical proofs of process equality.
It is also the basis for game constructions using Defender's Forcing.

\section{Finite Branching Property}\label{sec-Finite-Branching-Property}

Generally bisimilarity is undecidable for models with infinite branching transitions.
For the branching bisimilarity the finite branching property can be described by the following statement:
\begin{quote}
For each $P$ there is a finite set of processes $\{P_{i}\}_{i\in I}$ such that whenever $P\rightarrow^{*}\stackrel{\ell}{\longrightarrow}P'$ there is some $i\in I$ such that $P'\simeq P_{i}$.
\end{quote}
We prove in this section that both $\mathrm{nPDA}^{\epsilon+}$ and $\mathrm{PDA}^{\epsilon-}$ enjoy the finite branching property.
Let's take a look at the former first.

\begin{lemma}\label{2013-07-13}
In $\mathrm{nPDA}^{\epsilon+}$, $\min|pXC|\le\min\|pXC\|$ holds for all $pXC$.
\end{lemma}
\begin{proof}
Only external action can remove a symbol from an $\mathrm{nPDA}^{\epsilon+}$ process.
\qed \end{proof}

Using the simple property stated in Lemma~\ref{2013-07-13}, one can show that there is a constant bound for the length of the state-preserving transitions in $\mathrm{nPDA}^{\epsilon+}$.

\begin{lemma}\label{2013-07-15}
If $qXC\rightarrow q_{1}\beta_{1}C\rightarrow\ldots\rightarrow q_{k}\beta_{k}C$ for an $\mathrm{nPDA}^{\epsilon+}$ process $qXC$,
then $k<\mathfrak{q}\mathfrak{n}\mathfrak{r}(\mathfrak{m}+1)^{\mathfrak{q}}$.
\end{lemma}
\begin{proof}
Now suppose $qXC\rightarrow q_{1}Z_{1}\delta_{1}C$.
Let $k_{1}=\min\|q_{1}Z_{1}\delta_{1}C\|$ and let
\[q_{1}Z_{1}\delta_{1}C\rightarrow^{*}\stackrel{\jmath_{1}^{1}}{\longrightarrow}\ldots\rightarrow^{*}\stackrel{\jmath_{j_{1}}^{1}}{\longrightarrow}\rightarrow^{*}r_{1}C
\rightarrow^{*}\stackrel{\jmath_{j_{1}+1}^{1}}{\longrightarrow}\ldots\rightarrow^{*}\stackrel{\jmath_{j_{k_{1}}}^{1}}{\longrightarrow}\rightarrow^{*}h_{1}\]
be a transition sequence of minimal length that empties the stack.
Clearly $j_{1}\le \mathfrak{r}\mathfrak{m}$.
Now suppose $q_{1}Z_{1}\delta_{1}C\rightarrow^{*} q_{2}Z_{2}\delta_{2}\delta_{1}C$ such that
\begin{equation}\label{2013-07-18}
\mathfrak{r}\mathfrak{m}<|Z_{2}\delta_{2}\delta_{1}|\le \mathfrak{r}(\mathfrak{m}+1).
\end{equation}
Let $k_{2}=\min\|q_{2}Z_{2}\delta_{2}\delta_{1}C\|$ and let
\[q_{2}Z_{2}\delta_{2}\delta_{1}C\rightarrow^{*}\stackrel{\jmath_{1}^{2}}{\longrightarrow}\ldots\rightarrow^{*}\stackrel{\jmath_{j_{2}}^{2}}{\longrightarrow}\rightarrow^{*}r_{2}C
\rightarrow^{*}\stackrel{\jmath_{j_{2}+1}^{2}}{\longrightarrow}\ldots\rightarrow^{*}\stackrel{\jmath_{j_{k_{2}}}^{2}}{\longrightarrow}\rightarrow^{*}h_{2}\]
be a transition sequence of minimal length that empties the stack.
One must have $j_{2}>j_{1}$ according to (\ref{2013-07-18}).
By iterating the above argument one gets from
\begin{eqnarray*}
q_{1}Z_{1}\delta_{1}C &\rightarrow^{*}& q_{2}Z_{2}\delta_{2}\delta_{1}C \\
 &\rightarrow^{*}& \ldots \\
 &\rightarrow^{*}& q_{i+1}Z_{i+1}\delta_{i+1}\delta_{i}\ldots\delta_{1}C \\
 &\rightarrow^{*}& \ldots \\
 &\rightarrow^{*}& q_{\mathfrak{q}+1}Z_{\mathfrak{q}+1}\delta_{\mathfrak{q}+1}\delta_{\mathfrak{q}}\ldots\delta_{1}C
\end{eqnarray*}
with $\mathfrak{r}\mathfrak{m}(\mathfrak{m}+1)^{i-1}<|Z_{i+1}\delta_{i+1}\delta_{i}\ldots\delta_{1}|\le \mathfrak{r}(\mathfrak{m}+1)^{i}$ for all $i\in[\mathfrak{q}]$, some states $r_{1},\ldots,r_{\mathfrak{q}+1}$, some numbers $k_{1}<\ldots<k_{\mathfrak{q}+1}$ and $h_{1},\ldots,h_{\mathfrak{q}+1}$.
For each $i\in[\mathfrak{q}+1]$ there is some transition sequence
\[q_{i}Z_{i}\delta_{i}\ldots\delta_{1}C\rightarrow^{*}\stackrel{\jmath_{1}^{i}}{\longrightarrow}\ldots\rightarrow^{*}\stackrel{\jmath_{j_{i}}^{i}}{\longrightarrow}\rightarrow^{*}r_{i}C
\rightarrow^{*}\stackrel{\jmath_{j_{i}+1}^{i}}{\longrightarrow}\ldots\rightarrow^{*}\stackrel{\jmath_{j_{k_{i}}}^{i}}{\longrightarrow}\rightarrow^{*}h_{i}\]
where $k_{i}=\min\|q_{i}Z_{i}\delta_{i}\ldots\delta_{1}C\|$.
Since there are only $\mathfrak{q}$ states, there must be some $t_{1},t_{2}$ such that $0<t_{1}<t_{2}\le \mathfrak{q}+1$ and $r_{t_{1}}=r_{t_{2}}$.
It follows from the minimality that $j_{k_{t_{1}}}-j_{t_{1}}=j_{k_{t_{2}}}-j_{t_{2}}$.
But $j_{t_{2}}>j_{t_{1}}$.
Consequently $j_{k_{t_{1}}}<j_{k_{t_{2}}}$.
This inequality contradicts to the fact that $q_{t_{1}}Z_{t_{1}}\delta_{t_{1}}\ldots\delta_{1}C\simeq q_{t_{2}}Z_{t_{2}}\delta_{t_{2}}\ldots\delta_{1}C$.
We conclude that if $q XC\rightarrow^{*}q'\gamma C$ then $|\gamma| < \mathfrak{r}(\mathfrak{m}+1)^{\mathfrak{q}}$.
It follows from our convention that $k<\mathfrak{q}\mathfrak{n}\mathfrak{r}(\mathfrak{m}+1)^{\mathfrak{q}}$.
\qed \end{proof}

A proof of the following corollary can be read off from the above proof.

\begin{corollary}\label{2013-07-20}
Suppose $P$ is an $\mathrm{nPDA}^{\epsilon+}$ process.
There is a computable bound on the size of any $\mathrm{nPDA}^{\epsilon+}$ process $p\alpha$ such that $p\alpha\simeq P$.
\end{corollary}

Using Lemma~\ref{2013-07-15} one can define for $\mathrm{nPDA}^{\epsilon+}$ the approximation relation $\simeq_{n}$, the branching bisimilarity up to depth $n$, in the standard fashion.
The infinite approximation sequence approaches to $\simeq$ in the sense that $\bigcap_{i\in\mathbb{N}}\simeq_{i}$ coincides with $\simeq$ on $\mathrm{nPDA}^{\epsilon+}$ processes.
The following theorem follows from the fact that $\simeq_{i}$ is decidable for all $i\in\mathbb{N}$.

\begin{theorem}\label{2013-07-16}
The relation $\not\simeq_{\mathrm{nPDA}^{\epsilon+}}$ is semidecidable.
\end{theorem}

In $\mathrm{PDA}^{\epsilon-}$ one could have equality like $pY\simeq qX^{n}Y$ where an action of $pY$ is bisimulated by a sequence of transitions whose length depends on the size of $qX^{n}Y$.
There is no constant bound for the length of a state preserving silent transitions in $\mathrm{PDA}^{\epsilon-}$.
However the finite branching property clearly holds for the normed $\mathrm{PDA}^{\epsilon-}$ processes.
For an unnormed $\mathrm{PDA}^{\epsilon-}$ process $pX$ notice that due to the restriction on the silent transitions and our convention, the silent transition sequences $pX$ can induce must be of the form $pX\stackrel{\epsilon}{\longrightarrow}q_{1}Y_{1}\stackrel{\epsilon}{\longrightarrow}\ldots\stackrel{\epsilon}{\longrightarrow}q_{m}Y_{m}$ with $m$ bounded by $\mathfrak{n}\mathfrak{q}$.
We conclude that all $\mathrm{PDA}^{\epsilon-}$ processes enjoy the finite branching property.

\begin{theorem}\label{2014-01-02}
The relation $\not\simeq_{\mathrm{PDA}^{\epsilon-}}$ is semidecidable.
\end{theorem}

\section{Decidability of $\mathrm{PDA}^{\epsilon-}$}\label{sec-Decidability of PDAminus}

The basic idea of our proof of the decidability of $\mathrm{PDA}^{\epsilon-}$ is drawn from Stirling's proof for the strong bisimilarity on PDA~\cite{Stirling1996-nPDA-decidability,Stirling1998-nPDA-decidability,Stirling2000-PDA-decidability}.
To explain the key technical tool of Stirling's proof, it is helpful to recall the proof of the semidecidability of the strong bisimilarity of BPA~\cite{HuttelStirling1991}.
To check if $X\alpha\sim Y\beta$, we decompose the goal $X\alpha=Y\beta$ into say subgoals $\alpha=\gamma\beta$ and $X\gamma\beta=Y\beta$ derivable from the bisimulation property.
The latter can be simplified to $X\gamma=Y$ by cancellation.
Now the size of $\gamma$ is small as it were because $\gamma$ is derived in a computationally bounded number of steps.
It follows that the subgoal $X\gamma=Y$ is small and the subgoal $\alpha=\gamma\beta$ is smaller than $X\alpha=Y\beta$ in the sense that $\alpha$ is smaller than $X\alpha$ and the size of $\gamma$ is under control.
Using this `smallness' property we can build a finite tree of subgoals, called a tableau, in an organized fashion.
A semidecidable procedure is then designed by enumerating all finite tableaux and checking if any one of them giving rise to a strong bisimulation.
The unnormed BPA processes enjoy the following weak cancelation property:
If there is an infinite family of pairwise nonbisimilar BPA processes $\{\delta_{i}\}_{i\in\mathbb{N}}$ such that $\alpha\delta_{i}\sim\beta\delta_{i}$ for all $i\in\mathbb{N}$, then $\alpha\sim\beta$.
This weak cancelation guarantees that in a tree of subgoals there cannot be a path containing an infinite number of subgoals $\{\alpha\delta_{i}\sim\beta\delta_{i}\}_{i\in\mathbb{N}}$, where $\alpha\not\sim\beta$, without producing equivalent subgoals.
It follows from K\"{o}nig that only finite tableaux need be considered.
So the same semidecidable procedure works for the unnormed BPA.

Given BPA processes $\alpha,\beta$ with $\alpha\not\sim\beta$, we say that $\{\gamma_i\}_{i\in I}$ is a {\em minimal set of fixpoints} for $\alpha,\beta$ if the following hold:
\begin{itemize}
\item
For each $i\in I$ the process $\gamma_i$ is a fixpoint for $\alpha,\beta$, i.e. $\alpha\gamma_{i}\sim\beta\gamma_{i}$.
\item
For all $i,j\in I$ if $i\ne j$ then $\gamma_i\not\sim\gamma_j$.
\item
If $\alpha\gamma\sim\beta\gamma$ then $\gamma\sim\gamma_{i}$ for some $i\in I$.
\end{itemize}
Both the strong and the weak cancellation properties of BPA can be reiterated in the following more enlightening manner.
\begin{lemma}\label{2014-02-25}
Let $\alpha,\beta$ be BPA processes.
If $\alpha\not\sim\beta$ then the minimal set of fixpoints for $\alpha,\beta$ is finite.
\end{lemma}
The property described in Lemma~\ref{2014-02-25}, called the finite representation property in this paper, is the prime reason for the semidecidability of $\sim$ on BPA.
Stirling's remarkable observation is that the property described in Lemma~\ref{2014-02-25} is also valid for the strong bisimilarity on PDA.
What is subtle about PDA is that the fixpoints are stacks rather than processes due to the nonstructural definition of PDA processes.
In fact they must be extended stacks if they are able to code up recursive behaviours.
Hence the recursive constants.

What we will prove in this section is that the property described in Lemma~\ref{2014-02-25} continues to be valid for the branching bisimilarity on $\mathrm{PDA}^{\epsilon-}$ and that the cancellation property stated in the lemma is sufficient for us to design a semidecidable procedure for the equivalence checking problem.

\subsection{Decomposition}\label{sec-Decomposition}

Given a $\mathrm{PDA}^{\epsilon-}$ process $P$ there is no bound on the size of a $\mathrm{PDA}^{\epsilon-}$ process $Q$ that is branching bisimilar to $P$.
Fortunately there is a special class of processes which we can use to bypass the problem.
The following definition applies to all PDA processes.
\begin{definition}
A simple process $P$ is {\em normal in} $PC$, notation $P\vartriangleleft PC$, if
\begin{itemize}
\item either $P=k$ for some $k\in \mathbb{N}$;
\item or $P=pX(P_{1},\ldots,P_{n})$ such that, for each $i\in[n]$, $P_{i}$ is normal in $P_{i}C$ and it does not hold that $pX(P_{1},\ldots,P_{n})C\rightarrow^{*}P_{i}C$.
\end{itemize}
A constant $(P_{1},\ldots,P_{n})$ is normal in $D$ if $P_{i}\vartriangleleft P_{i}D$  for each $i\in[n]$.
A simple process/constant is {\em normal} if it is normal in itself.
A recursive constant $V_{[n]}$ defined by $V_{[n]}=(P_{1},\ldots,P_{n})V_{[n]}$ is normal if $P_{i}\vartriangleleft P_{i}V_{[n]}$ for all $i\in[n]$.
\end{definition}

A typical application scenario of normal process is the following: Given a process $qYB$, convert $B$ into $DD'$ where $D\vartriangleleft DD'$ and $|D| =\mathfrak{m}$.
If $pXA\simeq qYB$ then we can decompose $pXA$ against the decomposition $qYDD'$ in such a way that $D'$ becomes a common suffix of a pair of bisimilar processes.
Before stating the next lemma, we remark that it is decidable to check if $\exists h.P\simeq h$.
Consequently we only need to focus on $pXA$ such that $\forall k.pXA\not\simeq k$.

\begin{lemma}\label{2013-07-31}
Let $pXA,qYB$ be $\mathrm{PDA}^{\epsilon-}$ processes such that $\forall h.pXA\not\simeq h\not\simeq qYB$.
If $pXA\simeq qYB$, there are simple constants $U_{[\mathfrak{q}]}^{l},U_{[\mathfrak{q}]}^{r}$ and constant $D$ such that
\begin{enumerate}
\item $qYU_{[\mathfrak{q}]}^{r}D\simeq qYB$ and $U_{[\mathfrak{q}]}^{r}=(U_{[\mathfrak{q}]}^{r}D){\upharpoonright_{\mathfrak{m}}}$ and $U_{\mathfrak{q}}^{r} \vartriangleleft U_{\mathfrak{q}}^{r}D$;
\item $pXA\simeq pXU_{[\mathfrak{q}]}^{l}D$ and $|U_{[\mathfrak{q}]}^{l}|\le \mathfrak{r}\mathfrak{m}+1$.
\end{enumerate}
\end{lemma}
\begin{proof}
1. We explain how to decompose $B$ into $U_{\mathfrak{q}}^{r}D$ with $U_{\mathfrak{q}}^{r} \vartriangleleft U_{\mathfrak{q}}^{r}D$.
Firstly we construct a constant $U_{[\mathfrak{q}]}^{1}$ in the following fashion:
Suppose $B$ is an $m$-dimensional constant.
\begin{itemize}
\item
If $i\in[\mathfrak{q}]$ and $i\le m$ and $B(i)\rightarrow^{*}{\bf 0}$, let $U_{[\mathfrak{q}]}^{1}(i)={\bf 0}$.
\item
If $i\in[\mathfrak{q}]$ and $i>m$, let $U_{[\mathfrak{q}]}^{1}(i)=i$.
\item
If $i\in[\mathfrak{q}]$ and $i\le m$ and $B(i)\rightarrow^{*}k$, let $U_{[\mathfrak{q}]}^{1}(i)=k$.
\item
If $i\in[\mathfrak{q}]$ and $i\le m$ and $B(i)\rightarrow^{*}q_{i}Y_{i}B^{i}\nrightarrow$ for some $q_{i}Y_{i}B^{i}$, let $U_{[\mathfrak{q}]}^{1}(i)=q_{i}Y_{i}B^{i}$.
\end{itemize}
The constant $U_{[\mathfrak{q}]}^{1}$ is not unique.
But this is not a problem.
For each component of $U_{[\mathfrak{q}]}^{1}$ that is of the shape $q'Y'B'$ we continue to apply the decomposition.
This procedure will eventually stop if no recursive constant ever occurs.

For a recursive constant $V_{[n]}=(P_1,\ldots,P_n)V_{[n]}$ in $B$ we carry out the following transformation.
\begin{itemize}
\item [(i)] If there is a circular silent transition sequence of the form $P_{j_0}V_{[n]}\Longrightarrow P_{j_1}V_{[n]}\Longrightarrow\ldots\Longrightarrow P_{j_m}V_{[n]}\Longrightarrow P_{j_0}V_{[n]}$, then let $P_{j_0}'={\bf 0}$, $P_{j_1}'={\bf 0}$, \ldots, $P_{j_m}'={\bf 0}$.
\item [(ii)] After all circular silent transition cycles have been dealt with, we apply the above construction to transform $P_jV_{[n]}$, for each remaining $P_j$, to some $P_j'V_{[n]}$ such that either $P_j'={\bf 0}$ or $P_j'\vartriangleleft P_j'V_{[n]}$ with $|P_j'|>0$.
    This construction must terminate since $P_j$ is a simple process.
\end{itemize}
In the end we get a simple constant $(\overline{P}_1,\ldots,\overline{P}_n)$ where $\overline{P}_j$ is $j$ if $P_j'={\bf 0}$ and is $P_j'$ otherwise.
It is clear that $\overline{P}_j\vartriangleleft \overline{P}_jV_{[n]}$ and $V_{[n]}(j)\simeq\overline{P}_jV_{[n]}$ for all $j\in[n]$.
We will say that $\overline{P}_jV_{[n]}$ is a {\em normal unfolding} of $V_{[n]}$ at $j$ and $\{\overline{P}_jV_{[n]}\}_{j\in[n]}$ is a {\em normal unfolding} of $V_{[n]}$.

Our overall construction terminates when we reach either ${\bf 0}$ or a selection process or some $V_{[n]}=(P_1,\ldots,P_n)V_{[n]}$ for which we have already derived the normal unfolding.
Furthermore we need to make sure that normal unfolding has been applied for just enough number of times so that one can apply the $(\_){\upharpoonright_{\mathfrak{m}}}$ operation.
Let $U_{[\mathfrak{q}]}$ be the final constant obtained after all the decomposition.
According to our construction the constant $U_{[\mathfrak{q}]}$ satisfies $qYB\simeq qYU_{[\mathfrak{q}]}$.
Now decompose $U_{[\mathfrak{q}]}$ into $U_{[\mathfrak{q}]}{\upharpoonright_{\mathfrak{m}}}\cdot D$ and let $U_{[\mathfrak{q}]}^{r}=U_{[\mathfrak{q}]}{\upharpoonright_{\mathfrak{m}}}$.

2. Suppose $pXA\simeq qYB$.
Let
\begin{equation}\label{2013-08-07}
pXA\rightarrow^{*}\stackrel{\jmath_{1}}{\longrightarrow}\rightarrow^{*}\ldots\rightarrow^{*}\stackrel{\jmath_{k}}{\longrightarrow}\rightarrow^{*}A(h)
\end{equation}
be a sequence reaching $A(h)$ with minimal $k$.
Since $\simeq$ is closed under composition, $k$ cannot be greater than $\mathfrak{m}$.
It follows from $\forall h.pXA\not\simeq h$ that $k>0$.
The action sequence (\ref{2013-08-07}) must be bisimulated by $qYU_{[\mathfrak{q}]}^{r}D$ in the following manner:
\begin{equation}\label{2014-02-19}
qYU_{[\mathfrak{q}]}^{r}D\rightarrow^{*}\stackrel{\jmath_{1}}{\longrightarrow}Q_{1}D\rightarrow^{*}\stackrel{\jmath_{2}}{\longrightarrow}Q_{2}D\ldots\rightarrow^{*}\stackrel{\jmath_{k}}{\longrightarrow}Q_{h}D.
\end{equation}
Since $U_{[\mathfrak{q}]}^{r}$ is thick enough as it were, the constant $D$ remains intact throughout the transitions in (\ref{2014-02-19}).
Moreover $|Q_{h}|\le \mathfrak{r}\mathfrak{m}+1$.
Let $U_{[\mathfrak{q}]}^{l}(h)=Q_{h}$.
For $i\notin {\rm def}\|pX\|$ we let $U_{[\mathfrak{q}]}^{l}(i)$ be ${\bf 0}$.
It is clear that $pXA\simeq pXU_{[\mathfrak{q}]}^{l}D$.
\qed\end{proof}

\subsection{Finite Representation}\label{sec-Finite-Representation}

We now establish for $\mathrm{PDA}^{\epsilon-}$ the finite representation property.
In the following lemma the equivalence (\ref{2014-02-19a}) is the fixpoint property while the equivalence (\ref{2014-02-19b}) is the minimality property.

\begin{lemma}\label{2013-08-01}
Suppose $P,Q$ are simple $\mathrm{PDA}^{\epsilon-}$ processes with $ln(P)\subseteq[n]\supseteq ln(Q)$.
A finite set of recursive constant $\left\{V_{[n]}^{k}=(L_{1}^{k},\ldots,L_{n}^{k})V_{[n]}^{k}\right\}_{k\in K}$ exists such that
\begin{eqnarray}
PV_{[n]}^{k} &\simeq& QV_{[n]}^{k} \label{2014-02-19a}
\end{eqnarray}
for all $k\in K$ and for each $D_{[n]}$ satisfying $PD_{[n]} \simeq QD_{[n]}$ there is some $k\in K$ rendering true the following equivalence.
\begin{eqnarray}
D_{[n]} &\simeq& (L_{1}^{k},\ldots,L_{n}^{k})D_{[n]}. \label{2014-02-19b}
\end{eqnarray}
\end{lemma}
\begin{proof}
Suppose $PD_{[n]} \simeq QD_{[n]}$.
We will construct $V_{[n]}^{k}$ by induction such that at each step (\ref{2014-02-19b}) is maintained.
Let $V^{0}$ be $I_{[n]}$.
Thus $V^{0}$ is defined by $V^{0} = (1,\ldots,n)V^{0}$.
The recursive constant $V^{0}$ trivially validates (\ref{2014-02-19b}).
If it also satisfies (\ref{2014-02-19a}), we are done.
Otherwise we refine $V^{0}$ to some $V^{1}$ by the following induction.
Suppose $V^{d}=(L_{1}^{d},\ldots,L_{n}^{d})V^{d}$ has been constructed such that
\begin{eqnarray}
PV^{d} &\not\simeq& QV^{d}, \nonumber \\
D_{[n]} &\simeq& (L_{1}^{d},\ldots,L_{n}^{d})D_{[n]}. \label{2014-02-20b}
\end{eqnarray}
Let $m$ be the least number such that $P V^{d}\not\simeq_{m}Q V^{d}$.
We refine $V^{d}$  to $V^{d+1}$ by exploring the mismatch between the following equality and inequality:
\begin{eqnarray}
PD_{[n]} &\simeq& QD_{[n]}, \label{2013-08-11-eq} \\
P V^{d} &\not\simeq_{m}& Q V^{d}. \label{2013-08-11-ineq}
\end{eqnarray}
It follows from (\ref{2013-08-11-ineq}) that some transition $PV^{d}\stackrel{\jmath}{\longrightarrow}P'V^{d}$ exists such that for all transition sequence $QV^{d}\stackrel{\epsilon}{\longrightarrow}O_{1}V^{d}\stackrel{\epsilon}{\longrightarrow}\ldots\stackrel{\epsilon}{\longrightarrow}O_{e}V^{d}\stackrel{\jmath}{\longrightarrow}O'V^{d}$ with $PV^{d}\simeq_{m-1}O_{1}V^{d}$, \ldots, $PV^{d}\simeq_{m-1}O_{e}V^{d}$ one must have
\begin{eqnarray}\label{2014-02-26a}
P'V^{d} &\not\simeq_{m-1}& O'V^{d}.
\end{eqnarray}
According to (\ref{2013-08-11-eq}) however the transition $PD_{[n]}\stackrel{\jmath}{\longrightarrow}P'D_{[n]}$ must be matched by some transition sequence $QD_{[n]}\stackrel{\epsilon}{\longrightarrow}Q_{1}D_{[n]}\stackrel{\epsilon}{\longrightarrow}\ldots\stackrel{\epsilon}{\longrightarrow}Q_{e'}D_{[n]}\stackrel{\jmath}{\longrightarrow}Q'D_{[n]}$ such that $PD_{[n]}\simeq Q_{1}D_{[n]}$, \ldots, $PD_{[n]}\simeq Q_{e'}D_{[n]}$ and
\begin{eqnarray}\label{2014-02-26b}
P'D_{[n]} &\simeq& Q'D_{[n]}.
\end{eqnarray}
A special instance of (\ref{2014-02-26a}) is
\begin{eqnarray}\label{2014-02-26c}
P'V^{d} &\not\simeq_{m-1}& Q'V^{d}.
\end{eqnarray}
The above construction takes us from (\ref{2013-08-11-eq},\ref{2013-08-11-ineq}) to (\ref{2014-02-26b},\ref{2014-02-26c}).
By repeating the construction, we eventually get the following equality and inequality for some $L$:
\begin{eqnarray}
D_{[n]}(i) &\simeq& LD_{[n]}, \label{2013-08-11-equa} \\
V^{d}(i) &\not\simeq_{m'}& LV^{d}. \label{2013-08-11-ineq-m}
\end{eqnarray}
We continue the proof by looking at the shape of $L$.
\begin{itemize}
\item
$L=j\in[n]$ and $V^{d}(j)={\bf 0}$.
If $i<j$ then let
\[V^{d+1} = (L_{1}^{d},\ldots,L_{j-1}^{d},i,L_{j+1}^{d},\ldots,L_{n}^{d})V^{d+1}.\]
Otherwise let
\[V^{d+1} = (L_{1}^{d},\ldots,L_{i-1}^{d},j,L_{i+1}^{d},\ldots,L_{n}^{d})V^{d+1}.\]
Clearly $V^{d+1}$ validates (\ref{2014-02-20b}).
Notice that $V^{d+1}$ slightly violates the requirement on recursive constant.
This will be rectified in the end.
\item
$|L|>0$.
In this case let
$V^{d+1} = (L_{1}^{d},\ldots,L_{i-1}^{d},L,\overline{L_{i+1}^{d}},\ldots,\overline{L_{n}^{d}})V^{d+1}$, where for each $j\in\{i+1,\ldots,n\}$ the process $\overline{L_{j}^{d}}$ is defined as follows: If there are $j_{1}>\ldots>j_{g}>i$ such that $V^{d}(j)=j_{1}$, $V^{d}(j_{1})=j_{2}$, \ldots, $V^{d}(j_{g})=i$, then $\overline{L_{j}^{d}}=L$; otherwise $\overline{L_{j}^{d}}=L_{j}^{d}$.
Again $V^{d+1}$ trivially validates the equivalence (\ref{2014-02-20b}).
\end{itemize}
The construction must stop after at most $\frac{n(n-1)}{2}$ steps.
Eventually we get some $V=(L_{1},\ldots,L_{n})V$.
Modify the definition of $V$ as follows: For each $i\in[n]$ let $V(i)=i$ if $V(i)$ is a number.
What we get is the required $V_{[n]}^{k}$.
Starting with $I_{[n]}$ there are only finitely many such $V_{[n]}^{k}$ one can construct in $\frac{n(n-1)}{2}$ steps due to the finite branching property.
We are done.
\qed\end{proof}

Lemma~\ref{2013-07-31} allows one to create common suffix by introducing a normal constant, whereas Lemma~\ref{2013-08-01} helps to substitute a recursive constant for the suffix.
We get a more useful result if we combine these two lemmas.

\begin{lemma}\label{2014-02-20-lemma}
Fix simple processes $pX,qYU^{r}$, where $U^{r}$ is normal and $|U^{r}| =\mathfrak{m}$ and $ln(qYU^{r})\subseteq[n]$.
A finite family $\left\{V_{[n]}^{k}=(L_{1}^{k},\ldots,L_{n}^{k})V_{[n]}^{k}\right\}_{k\in K}$ of recursive constant exists such that for all $A$ with $ln(A)\subseteq[n]$ and all $D_{[n]}$ with $U^{r}\vartriangleleft U^{r}D_{[n]}$ and $pXAD_{[n]} \simeq qYU^{r}D_{[n]}$, there is some $k\in K$ rendering true the following.
\begin{eqnarray}
pXAV_{[n]}^{k} &\simeq& qYU^{r}V_{[n]}^{k}, \label{2014-02-21a}\\
D_{[n]} &\simeq& (L_{1}^{k},\ldots,L_{n}^{k})D_{[n]}. \label{2014-02-21b}
\end{eqnarray}
\end{lemma}
\begin{proof}
By the proof of Lemma~\ref{2013-07-31} there is a finite set $\left\{U^{j}=(G_{1}^{j},\ldots,G_{\mathfrak{q}}^{j})\right\}_{j\in J}$ such that for each pair $A,D_{[n]}$ with $ln(A)\subseteq[n]$ and $pXAD_{[n]} \simeq qYU^{r}D_{[n]}$ there is some $U^{j}=(G_{1}^{j},\ldots,G_{\mathfrak{q}}^{j})$ validating the following.
\begin{eqnarray}
pXU^{j}D_{[n]} &\simeq& qYU^{r}D_{[n]}, \label{boston-1} \\
A(i)D_{[n]} &\simeq& G_{i}^{j}D_{[n]}\ \mathrm{for}\ \mathrm{all}\ i\in[n]. \label{boston-2}
\end{eqnarray}
For each $j\in J$ let $\nabla^{j}$ be the set of pairs $A,D_{[n]}$ that satisfy $ln(A)\subseteq[n]$ and $U^{r}\vartriangleleft U^{r}D_{[n]}$ and $pXAD_{[n]} \simeq qYU^{r}D_{[n]}$ and (\ref{boston-1}) and (\ref{boston-2}).
It follows from (\ref{boston-1} ) and Lemma~\ref{2013-08-01} that there is a finite family of recursive constants $\left\{V_{[n]}^{k}=(L_{1}^{k},\ldots,L_{n}^{k})V_{[n]}^{k}\right\}_{k\in K}$ such that for each pair $A,D_{[n]}$ in $\nabla^{j}$ there is some $k\in K$ rendering true the following.
\begin{eqnarray}
pXU^{j}V_{[n]}^{k} &\simeq& qYU^{r}V_{[n]}^{k}, \label{boston-3} \\
D_{[n]} &\simeq& (L_{1}^{k},\ldots,L_{n}^{k})D_{[n]}. \label{boston-4}
\end{eqnarray}
We now prove the following equivalence using a game theoretical argument.
\begin{eqnarray}
pXAV_{[n]}^{k} &\simeq& qYU^{r}V_{[n]}^{k}. \label{boston-5}
\end{eqnarray}
So suppose the bisimulation game starts with the configuration
\begin{equation} \label{2014-04-02}
(pXAV_{[n]}^{k}, qYU^{r}V_{[n]}^{k}).
\end{equation}
No matter how Attacker plays the Defender of the game (\ref{2014-04-02}) mimics the Defender of the game (\ref{boston-3}), who in turn mimics the Defender of the game (\ref{boston-1}).
If a configuration of the following form
\begin{equation} \label{boston-6}
(A(i)V_{[n]}^{k}, PV_{[n]}^{k})
\end{equation}
is reached, then the game of (\ref{boston-3}) would reach the configuration
\begin{eqnarray}
G_i^j V_{[n]}^{k} &\simeq& PV_{[n]}^{k}, \label{boston-7}
\end{eqnarray}
and the game of (\ref{boston-1}) would reach the configuration
\begin{eqnarray}
G_i^j D_{[n]} &\simeq& PD_{[n]}. \label{boston-8}
\end{eqnarray}
It follows from (\ref{boston-2}) and (\ref{boston-8}) that
\begin{eqnarray}
A(i)D_{[n]} &\simeq& PD_{[n]}. \label{boston-9}
\end{eqnarray}
Now the Defender of the game (\ref{boston-6}) simply copycats the Defender's strategy of the game (\ref{boston-9}), invoking the Defender's strategy of the game (\ref{boston-4}) whenever necessary.
What we have described is a winning strategy for the Defender of the game (\ref{2014-04-02}).
We conclude that (\ref{boston-5}) is valid.
\qed\end{proof}

What Lemma~\ref{2014-02-20-lemma} says is that the order of an application of Lemma~\ref{2013-07-31} followed by an immediate application of Lemma~\ref{2013-08-01} can be swapped without sacrificing the finite representation property.
This reordering would be very convenient in guaranteeing the termination of tableau construction.

\subsection{Tableau System}\label{sec-Tableau-System}

A straightforward way to prove bisimilarity between two processes is to construct a finite binary relation containing the pair of processes that can be extended to a bisimulation.
Such a finite relation is called a bisimulation base, originally due to Caucal~\cite{BurkartCaucalMollerSteffen2001}.
The tableau approach~\cite{HuttelStirling1991,Huttel1992} can be seen as an effective way of generating a bisimulation base.
Lemma~\ref{2013-07-31} and Lemma~\ref{2013-08-01} suggest the first three tableau rules for $\mathrm{PDA}^{\epsilon-}$ given in Figure~\ref{Tableau-Rules-4-Popping}.
To define the fourth tableau rule, we need the notion of {\em match}.
A match for an equality $P=Q$ is a {\em finite} set $\{P_{i}=Q_{i}\}_{i=1}^{k}$ containing those and only those equalities accounted for in the following statements:
\begin{enumerate}
\item
For each transition $P\stackrel{\ell}{\longrightarrow}P'$, one of the following holds:
\begin{itemize}
\item $\ell=\epsilon$ and $P'=Q\in\{P_{i}=Q_{i}\}_{i=1}^{k}$;
\item
there is a transition sequence $Q\stackrel{\epsilon}{\longrightarrow}Q_{1}\stackrel{\epsilon}{\longrightarrow}\ldots\stackrel{\epsilon}{\longrightarrow}Q_{n}\stackrel{\ell}{\longrightarrow}Q'$ such that $\{P=Q_{1},\ldots,P=Q_{n},P'=Q'\}\subseteq\{P_{i}=Q_{i}\}_{i=1}^{k}$.
\end{itemize}

\item
For each transition $Q\stackrel{\ell}{\longrightarrow}Q'$, one of the following holds:
\begin{itemize}
\item $\ell=\epsilon$ and $P=Q'\in\{P_{i}=Q_{i}\}_{i=1}^{k}$;
\item
there is a transition sequence $P\stackrel{\epsilon}{\longrightarrow}P_{1}\stackrel{\epsilon}{\longrightarrow}\ldots\stackrel{\epsilon}{\longrightarrow}P_{n}\stackrel{\ell}{\longrightarrow}P'$ such that $\{P_{1}=Q,\ldots,P_{n}=Q,P'=Q'\}\subseteq\{P_{i}=Q_{i}\}_{i=1}^{k}$.
\end{itemize}
\end{enumerate}
We remark that a match could be empty.
An empty match is a {\em successful match}.

\begin{figure*}[t]
\begin{center}
\begin{tabular}{|l|} \hline
 \\
Rdcp$^{\epsilon-}$ \ \ \
$\inference{rXA=sYB}{\left\{o_{j}Z_{j}N_{j}=N_{j}(k_{j})\right\}_{j\in J}\ \ \ \
rXA=sYU^{r}D}$ \
$\begin{array}{l}
U^{r}\vartriangleleft U^{r}D\ \mathrm{and}\ |U^{r}|\le\mathfrak{m}; \\
o_{j}Z_{j}N_{j}\stackrel{\epsilon}{\longrightarrow}N_{j}(k_{j})\ \mathrm{for}\ \mathrm{all}\ j\in J; \\
U^{r}D=B\ \mathrm{is}\ \mathrm{deducible}\ \mathrm{from} \ \left\{o_{j}Z_{j}N_{j}=N_{j}(k_{j})\right\}_{j\in J}; \\
\left\{o_{j}Z_{j}N_{j}=N_{j}(k_{j})\right\}_{j\in J} \ \mathrm{supports}\ \mathrm{a}\ \mathrm{normal}\ \mathrm{unfolding}\\
\mathrm{of}\ \mathrm{every}\ \mathrm{recursive}\ \mathrm{constant}\ \mathrm{appearing}\ \mathrm{in}\ B.
\end{array}$ \\
 \\
Ldcp$^{\epsilon-}$ \ \ \
$\inference{rXA=sYU^{r}D}{\left\{A(i)=G_{i}D\right\}_{i\in[\mathfrak{q}]}\ \ \ \
rXU^{l}D=sYU^{r}D}$ \ $U^{l}=\left(G_{i}\right)_{i\in[\mathfrak{q}]},\ |U^{l}|\le \mathfrak{r}\mathfrak{m}+1$. \\
 \\
Cancel$^{\epsilon-}$ \ \ \
$\inference{rXA D_{[n]}=sYBD_{[n]}}{\left\{L_{i}D_{[n]}=D_{[n]}(i)\right\}_{i\in[n]}\ \ \ \  rXAV_{[n]}=sYBV_{[n]}}$ \
$V_{[n]}=(L_{i})_{i\in[n]}V_{[n]}$. \\
 \\
Match$^{\epsilon-}$ \ \ \ \ \
$\inference{P=Q}{P_{1}=Q_{1}\ \ldots\ P_{l}=Q_{l}}$ \ $\{P_{1}=Q_{1},\ldots,P_{l}=Q_{l}\}\ \mathrm{is} \ \mathrm{a}\ \mathrm{match}\ \mathrm{for}\ P=Q$. \\
 \\
\hline
\end{tabular}
\end{center}
\caption{Tableau Rules for $\mathrm{PDA}^{\epsilon-}$ \label{Tableau-Rules-4-Popping}}
\end{figure*}

The rule Rdcp$^{\epsilon-}$ produces a decomposition of the right hand side of the goal.
A crucial point about this rule is that its side condition is semidecidable.
This is because to check $U^{r}\vartriangleleft U^{r}D$, one has to invoke a semidecidable procedure for $\not\simeq$.
The rule Ldcp$^{\epsilon-}$ decomposes the left hand side of the goal into a form that matches the right hand side, creating a common suffix $D$.
The Cancel$^{\epsilon-}$ rule harnesses the complexity by introducing a recursive subgoal.
The Match$^{\epsilon-}$ rule is standard.
All the rules are backward sound.

\subsection{Subtableau}

A subtableau is a building block for tableau.
Its chief role is to reduce a goal to a finite number of subgoals of controllable size.
A subtableau is manufactured using a strategy that applies Rdcp$^{\epsilon-}$, Ldcp$^{\epsilon-}$ and Cancel$^{\epsilon-}$ in an orderly manner.
Notice that Match$^{\epsilon-}$ is never used in the construction of any subtableaux.
The strategy is described as follows (confer Fig.~\ref{Subtableau-Construction-4-Popping}).

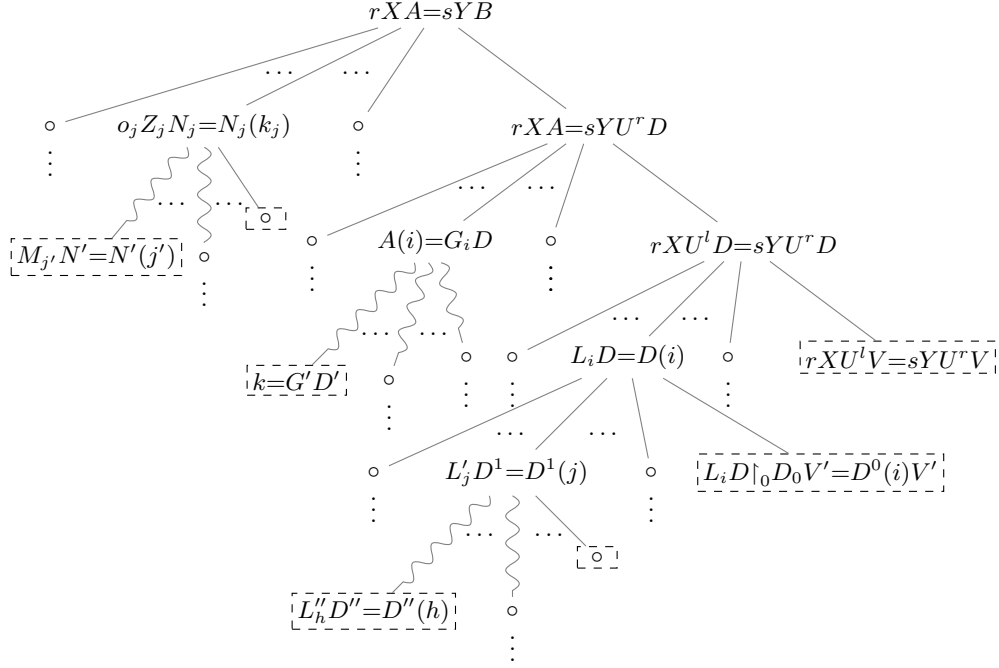
\begin{figure*}[t]
\usetikzlibrary{decorations.pathmorphing}
\begin{center}
   \resizebox{.8\textwidth}{!}{
    \begin{tikzpicture}
      \tikzstyle{leaf node}=[draw=black, dashed, inner sep=.5ex]
      \tikzstyle{ss}=[gray, -]
      \node(11) at (5,0) {$rXA{=}sYB\  $};
      \node(21) at (2,-1.5) {$o_jZ_jN_j{=}N_j(k_j)$};
      \node(22) at (7,-1.5) {$rXA{=}sYU^{r}D$};
      \node (c21) at (0,-1.5) {$\circ$};
      \node [below] at (c21) {$\vdots$};
      \node (c22) at (4,-1.5) {$\circ$};
      \node [below] at (c22) {$\vdots$};
      \node (dots21) at (3, -.8){$\dots$};
      \node (dots22) at (4, -.8) {$\dots$};

      \node(31) at (5,-3) {$ A(i){=}G_iD $};
      \node(32) at (9,-3) {$rXU^{l}D{=}sYU^{r}D$};
      \node(c33) at (3.4,-3) {$\circ$};
      \node(c34) at (6.5,-3) {$\circ$};
      \node [below] at (c33) {$\vdots$};
      \node [below] at (c34) {$\vdots$};

      \node[leaf node] (c31) at(.6,-3.2) {$M_{j'}N'{=}N'(j')$};
      \node(dots32) at (5.5,-2.3) {$\dots$};
      \node(dots33) at (6.4, -2.3) {$\dots$};
      \node (lf31) at (2,-3.2) {$\circ$};
      \node [below] at (lf31) {$\vdots$};
      \node [leaf node] (lf32) at (2.8, -2.7) {$\  \circ \ $};
      \node (dots_31) at (1.6, -2.5) {$\dots$};
      \node (dots_32) at (2.35, -2.5) {$\dots$};
      \node(41) at (7.5,-4.5) {$ L_iD {=} D(i)$};
      \node[leaf node] (42) at (11, -4.5) {{$rXU^{l}V {=} sYU^{r}V$}};
      \node(c43) at (6,-4.5) {$\circ$};
      \node(c44) at (8.8,-4.5) {$\circ$};
      \node [below] at (c43) {$\vdots$};
      \node [below] at (c44) {$\vdots$};

      \node (c41) at (4.4,-4.8) {$\circ$};
      \node(c42) at (5.4,-4.5) {$\circ$};
      \node[leaf node] (lf41) at (3.2,-4.8) {$k{=}G'D'$};
      \node(dots42) at (7.5, -4) {$\dots$};
      \node(dots43) at (8.4, -4) {$\dots$};
      \node [below] at (c41) {$\vdots$};
      \node [below] at (c42) {$\vdots$};
      \node (dots_41) at (4.25,-4.2) {$\dots$};
      \node (dots_42) at (5,-4.2) {$\dots$};
      \node(51) at (6,-6) {$ \ L_j'D^{1}{=} D^{1}(j)  $};
      \node[leaf node] (52) at (10, -6) {{$L_iD{\upharpoonright}_{0}D_0V'{=}D^{0}(i)V'$}};
      \node(c51) at (4.2, -6) {$\circ$};
      \node(c52) at (7.8,-6) {$\circ$};
      \node(dots51) at (6, -5.5) {$\dots$};
      \node(dots52) at (7.2, -5.5) {$\dots$};
      \node [leaf node] (c61) at (4.2,-7.8) {$L''_hD''{=}D''(h)$};
      \node  (c62) at (6,-7.8) {$\circ$};
      \node [below] at (c62) {$\vdots$};
      \node [leaf node] (c64) at (7.1, -7.1) {$\  \circ \  $};
      \node (dots61) at (5.6, -6.8) {$\dots$};
      \node (dots62) at (6.5, -6.8) {$\dots$};

      \path [ss] (11) edge (21);
      \path [ss] (11) edge (c21);
      \path [ss] (11) edge (c22);
      \path [ss] (11) edge (22);
      \path [ss] (22) edge (31);
      \path [ss] (22) edge (c33);
      \path [ss] (22) edge (c34);
      \path [ss] (22) edge (32);
      \draw[ss,decorate, decoration={coil,aspect=0} ] (21) -- (c31);
      \draw[ss] (21) -- (lf32);
      \draw[ss,decorate, decoration={coil,aspect=0} ] (21) -- (lf31);

      \path [ss] (32) edge (41);
      \path [ss] (32) edge (42);
      \draw[ss,decorate, decoration={coil,aspect=0} ] (31) -- (c41);
      \draw[ss,decorate,decoration={coil,aspect=0}] (31) -- (c42);
      \draw[ss,decorate, decoration={coil,aspect=0} ] (31) -- (lf41);

      \path [ss] (32) edge (c43);
      \path [ss] (32) edge (c44);

      \path [ss] (41) edge (51);
      \path [ss] (41) edge (52);
      \path [ss] (41) edge (c51);
      \path [ss] (41) edge (c52);

      \draw[ss,decorate,decoration={coil,aspect=0}]  (51) -- (c61);
      \draw[ss] (51) -- (c64);
      \draw[ss,decorate,decoration={coil,aspect=0}]  (51) -- (c62);

      \node [below] at (c34) {$\vdots$};
      \node [below] at (c51) {$\vdots$};
      \node [below] at (c52) {$\vdots$};

    \end{tikzpicture}}
  \end{center}
\caption{Subtableau Construction for $\mathrm{PDA}^{\epsilon-}$ \label{Subtableau-Construction-4-Popping}}
\end{figure*}

\begin{enumerate}
\item Given a nontrivial goal $rXA=sYB$, apply Rdcp$^{\epsilon-}$ first.
Two kinds of subgoals are produced.
\begin{enumerate}
\item
$o_{j}Z_{j}N_{j}=N_{j}(k_{j})$.
Go to Step 4.
\item
$rXA=sYU^{r}D$.
Go to Step 2.
\end{enumerate}

\item
Apply Ldcp$^{\epsilon-}$ to $rXA=sYU^{r}D$.
We get two classes of subgoals.
\begin{enumerate}
\item
$A(i)=G_{i}D$.
The size of $A(i)$ is strictly smaller than the size of $rXA$.
In this case go to Step 1 and carry out the construction inductively.
\item
$rXU^{l}D=sYU^{r}D$.
Go to Step 3.
\end{enumerate}

\item
Apply Cancel$^{\epsilon-}$ to $rXU^{l}D=sYU^{r}D$.
Two types of subgoals are generated.
\begin{enumerate}
\item
$L_{i}D=D(i)$.
Go to Step 4.
\item
$rXU^{l}V=sYU^{r}V$.
Since $|rXU^{l}|$ is bounded by $\mathfrak{r}\mathfrak{m}+2$ and $|sYU^{r}|$ is bounded by $\mathfrak{m}+1$, we take $rXU^{l}V=sYU^{r}V$ as a {\em leaf} of the subtableau.
\end{enumerate}

\item
$L_{i}D=D(i)$.
There are three subcases.
\begin{enumerate}
\item
If $L_{i}=i$, then $L_{i}D=D(i)$.
In this case $L_{i}D=D(i)$ is a {\em successful leaf}.
\item
If $D(i)$ is a simple process and $|D(i)|\le\mathfrak{m}$ and $L_{i}\ne i$, we take $L_{i}D=D(i)$ as a {\em leaf} of the subtableau.
\item
If $D(i)$ is a simple process and $|D(i)|>\mathfrak{m}$ and $L_{i}\ne i$, then $|D|>\mathfrak{m}+1$.
We cut $D$ at depth $1$, producing a decomposition $D=D{\upharpoonright}_{0}D'$.
We then cut $D'$ at the depth $\mathfrak{m}$, producing a decomposition $D'=(D'){\upharpoonright}_{\mathfrak{m}}D^{1}$.
Let $D^{0}=(D'){\upharpoonright}_{\mathfrak{m}}$.
Under this decomposition $D(i)=D^{0}(i)D^{1}$.
Now apply Cancel$^{\epsilon-}$ to $L_{i}D{\upharpoonright}_{0}D^{0}D^{1}=D^{0}(i)D^{1}$.
We get two types of subgoal.
\begin{enumerate}
\item
$L_{j}'D^{1}=D^{1}(j)$.
Repeat Step 4 inductively.
\item
$L_{i}D{\upharpoonright}_{0}D^{0}V'=D^{0}(i)V'$.
We take $L_{i}D{\upharpoonright}_{0}D^{0}V'=D^{0}(i)V'$ as a {\em leaf} since $D{\upharpoonright}_{0}D^{0}$ and $D^{0}(i)$ are bounded in size by $\mathfrak{m}+1$.
\end{enumerate}
\item If $D(i)$ contains recursive constants that cannot be decomposed away as it were, then apply normal unfolding to the the recursive constants if necessary for just an enough number of times so that we can apply the construction in 4(c) also in this case.
\end{enumerate}
\end{enumerate}
The key question about the construction of a subtableau using the above strategy is if it always terminates.
This is answered by the next lemma.

\begin{lemma}\label{2013-12-27}
Every subtableau is finite.
\end{lemma}
\begin{proof}
Step 4(c) cannot be repeated infinitely often.
This is because in the recursive invocation in 4(c)(i) the size of the right hand side of the subgoal has strictly decreased from $|D(i)|$ to $|D^{1}(j)|$.
Step 4(d) cannot be repeated infinitely often.
This is because in the recursive invocation in 4(d), similar to the recursive invocation in 4(c)(i), the right hand side of the subgoal must repeat for there are only a finite number of recursive constants in the right hand side of the root of the subtableau and the normal unfolding of each of them is fixed throughout the subtableau.
We get a {\em successful leaf} the first time a repetition occurs.
Step 2 can be repeated for only a finite number of time because of the the strict size decrease effect in case 2(a).
This implies that in the construction of the subtrees rooted at $A(i)=G_{i}D$, Step 3 is only invoked for a finite number of times.
It also implies that Step 1 can only be applied for a finite number of times.
\qed\end{proof}

\subsection{Tableau}

We are now in a position to explain how to produce a tableau for $rXA=sYB$.
To start with we construct the subtableau for $rXA=sYB$.
For each leaf of the subtableau that is neither successful nor unsuccessful, we try to apply Match$^{\epsilon+}$.
If it turns out that Match$^{\epsilon+}$ is not applicable, then the leaf is an unsuccessful leaf of the tableau.
If Match$^{\epsilon+}$ is applicable and the resulting match is empty, then the leaf is a successful leaf of the tableau; otherwise we repeat the subtableau construction for each subgoal of Match$^{\epsilon+}$.
In this way we obtain a {\em quasi tableau} for $rXA=sYB$.
The construction of a subtableau ends on a leaf $F$ of a subtableau if either it is a successful/unsuccessful leaf of the subtableau or it coincides with the leaf $F'$ of an ancestor subtableau with $F'$ staying in the path from the root of the quasi tableau to $F$.
In view of the coinductive nature of bisimulation we may think of $F$ as a {\em potentially successful} leaf.
The success of a potentially successful leaf depends on the (potential) success of other leaves.
\begin{definition}
A {\em tableau} is a quasi tableau in which only a finite number of recursive constants are introduced.
A tableau is {\em successful} if its leaves are either successful or potentially successful.
\end{definition}
A finite number of recursive constant must have a computable size bound.
Consequently all the leaves of all subtableaux appearing in a tableau have a computable size bound.
It follows that a leaf of a subtableau in a tableau is either successful/unsuccessful or is potentially successful.
Hence the following lemma.

\begin{lemma}\label{2014-12-22}
Every tableau is finite.
\end{lemma}

The above lemma is reassuring in that it guarantees that a tableau is either successful or unsuccessful.

\begin{lemma}\label{2013-08-12}
Suppose $\forall h\in\mathbb{N}.\,rXA\not\simeq h \not\simeq sYB$.
Then $rXA\simeq sYB$ if and only if $rXA=sYB$ has a successful tableau.
\end{lemma}
\begin{proof}
If $pXA\simeq q YB$, a successful tableau can be easily constructed using Lemma~\ref{2013-07-31}, Lemma~\ref{2013-08-01} and Lemma~\ref{2014-02-20-lemma}.
To prove the converse implication, assume that all the leaves of a successful tableau for $pXA=qYB$ are sound for $\simeq_{k}$.
The ancestor of a potentially successful leaf must be sound for $\simeq_{k+1}$ because all the rules are backward sound and there is at least one application of Match$^{\epsilon+}$ between a subtableau and its parent subtableau.
It follows from induction that all the equalities appearing in the tableau are sound for $\simeq_{k}$ for all $k$.
\qed\end{proof}

Lemma~\ref{2013-08-12} provides the following semidecidable procedure for checking $\simeq$ on $\mathrm{PDA}^{\epsilon-}$ processes:
Given input $rXA,sYB$, check if $rXA\simeq h$ or $qYB\simeq h$, for some $h\in\mathbb{N}$, or $rXA\simeq {\bf 0}$ or $qYB\simeq {\bf 0}$.
If the answer is positive, we use brutal force to check if $rXA\simeq sYB$.
Otherwise we enumerate all the tableaux for $rXA=sYB$ and at the same time check if any of them is successful.
This is possible because the side conditions of the rules used to generate a tableau are either decidable or semidecidable.
Together with Theorem~\ref{2014-01-02} we get the main result of the section.

\begin{theorem}\label{popping-PDA-decidability}
The branching bisimilarity on $\mathrm{PDA}^{\epsilon-}$ processes is decidable.
\end{theorem}

\section{Decidability of $\mathrm{nPDA}^{\epsilon+}$}\label{sec-Decidability-of-nPDAplus}

The decidability proof of Section~\ref{sec-Decidability of PDAminus} can be repeated for $\mathrm{nPDA}^{\epsilon+}$.
The difference is that $\mathrm{nPDA}^{\epsilon+}$ is simpler as far as equivalence checking is concerned.
Properties similar to those described in Lemma~\ref{2013-07-31}, Lemma~\ref{2013-08-01} and Lemma~\ref{2014-02-20-lemma} hold for $\mathrm{nPDA}^{\epsilon+}$.
We choose not to repeat the precise statements of these lemmas and their proofs.
What we are going to do is to present the tableau construction.
This would be informative enough for the reader to work out the details.

Lemma~\ref{2013-07-13} can be exploited to simplify the constructions of tableaux; it renders a rule like Rdcp$^{\epsilon-}$ unnecessary.
The tableau rules for $\mathrm{nPDA}^{\epsilon+}$ are defined in Figure~\ref{Tableau-Rules-4-Pushing}.
It follows from Lemma~\ref{2013-07-15} that the size of the $\mathfrak{q}$-ary simple constant $U$ introduced in Decmp$^{\epsilon+}$ is loosely bounded by $\mathfrak{q}\mathfrak{n}\mathfrak{r}^{2}(\mathfrak{m}+1)^{(\mathfrak{q}+1)}$ whenever $|B|\le\mathfrak{m}$.

\begin{figure*}[t]
\begin{center}
\begin{tabular}{|l|} \hline
 \\
Decmp$^{\epsilon+}$ \ \ \ \ \
$\inference{rXA=sYBD}{\{A(i)=G_{i}D\}_{i\;\in\;\mathrm{def}\,\|rX\|}\ \ \ \ \ rXUD=sYBD}$ \ $\begin{array}{l}
|B| =\mathfrak{m},\ |D|>0, \\
U=(G_{1},\ldots,G_{\mathfrak{q}}).
 \end{array}$ \\
 \\
Cancel$^{\epsilon+}$ \ \ \ \ \
$\inference{rXAD_{[n]}=sYBD_{[n]}}{\{L_{i}D_{[n]}=D_{[n]}(i)\}_{i\in[n]}\ \ \ \ \ rXAV_{[n]}=sYBV_{[n]}}$ \ $\begin{array}{l}
|B| =\mathfrak{m},\ |D|>0, \\
V_{[n]}=\left(L_{i}\right)_{i\in[n]}V_{[n]}.
 \end{array}$ \\
 \\
Match$^{\epsilon+}$ \ \ \ \ \
$\inference{P=Q}{P_{1}=Q_{1}\ \ldots\ P_{l}=Q_{l}}$ \ $\{P_{1}=Q_{1},\ldots,P_{l}=Q_{l}\}\ \mathrm{is} \ \mathrm{a}\ \mathrm{match}\ \mathrm{for}\ P=Q$. \\
 \\
\hline
\end{tabular}
\end{center}
\caption{Tableau Rules for $\mathrm{nPDA}^{\epsilon+}$. \label{Tableau-Rules-4-Pushing}}
\end{figure*}

We now explain how to construct subtableau for a pair of $\mathrm{nPDA}^{\epsilon+}$ processes.
Subtableaux for $\mathrm{nPDA}^{\epsilon+}$ are simpler than those for $\mathrm{PDA}^{\epsilon-}$.
Branching bisimilarity between an atomic process $k$ and any other process $P$ is easy to check algorithmically.
Now suppose the goal is to prove $rXA\simeq sYB$.
Without loss of generality we assume that $|rXA|\le|sYB|$.
If $|B|\le\mathfrak{m}$, the subtableau for $rXA=sYB$ is a single node tree labeled by $rXA=sYB$.
The subtableau for $rXA=sYBD$, where $|B| =\mathfrak{m}$ and $|D|>0$, is of the shape pictured in Fig.~\ref{Subtableau-Construction-4-Pushing}.
\begin{figure*}[t]
\begin{center}
\small
\[\inference{rXA=sYBD}{\ldots \inference{A(i)=G_{i}D}{\vdots}\ldots\ \ \ \inference{rXUD=sYBD}{\inference{\ldots\ \ \ \ \ \ \ L_{j}D=D(j)\ \ \ \ \ \ \ \ldots}{\ldots \inference{L_{j'}'D'=D'(j')}{\vdots}\ldots\ \ \ L_{j}B'V'=B'V'}\ \ \ rXUV=sYB V}}\]
\normalsize
\end{center}
\caption{Subtableau Construction for $\mathrm{nPDA}^{\epsilon+}$ \label{Subtableau-Construction-4-Pushing}}
\end{figure*}
It is generated inductively in the following fashion:
\begin{enumerate}
\item
Apply Decmp$^{\epsilon+}$ to $rXA=sYBD$.
If $A(i)=\epsilon$ or $G_{i}D=\epsilon$ then the subgoal $A(i)=G_{i}D$ of Decmp$^{\epsilon+}$ is a leaf.
It is successful whenever $A(i)\simeq G_{i}D$ and unsuccessful whenever $A(i)\not\simeq G_{i}D$.
If neither $A=\epsilon$ nor $\epsilon=G_{i}D$ then there are following subcases:
\begin{itemize}
\item
If $|G_{i}D|>\mathfrak{m}$, apply Decmp$^{\epsilon+}$ inductively to the subgoal $A(i)=G_{i}D$.
Notice that Decmp$^{\epsilon+}$ cannot be invoked infinitely often since the size of the left hand side of the subgoal strictly decreases.
\item
If $|G_{i}D|\le\mathfrak{m}$ and $|A|>\mathfrak{m}$, then swap the position of $A(i)$ and $G_{i}D$.
From this point onwards swapping will not happen again.
\item
If both $|G_{i}D|\le\mathfrak{m}$ and $|A|\le\mathfrak{m}$, then the subgoal $A(i)=G_{i}D$ is a leaf.
\end{itemize}
It should be evident that the leftmost path of the subtableau is finite.
\item
If the subgoal $rXUD=sBD$ of Decmp$^{\epsilon+}$ is small in the sense that $D=\epsilon$, then it is deemed as a leaf of the subtableau.
Otherwise we apply Cancel$^{\epsilon+}$ to the subgoal.
\item
The subgoal $rXUV=sB V$ of Cancel$^{\epsilon+}$ is of small size.
It is a leaf of the subtableau.
\item
For each subgoal $L_{j}D=D(j)$ of Cancel$^{\epsilon+}$, we take it as a leaf of the subtableau if the subgoal has a small size in the sense that $|D|\le \mathfrak{m}+1$.
Otherwise we let $D=B'D'$ with $|B'| =\mathfrak{m}+1$ and apply Cancel$^{\epsilon+}$ inductively to the subgoal.
This is well defined since the right hand side $L_{j'}'D'=D'(j')$ of the subgoal of the second application of Cancel$^{\epsilon+}$ is strictly smaller than the right hand side of $L_{j}D=D(j)$ of the first application.
\item
The subgoal $L_{j}B'V'=B'V'$ of the second application of Cancel$^{\epsilon+}$ is a leaf of the subtableau.
\end{enumerate}
It is easy to see from the above account that every subtableau is finite.

A tableau for $rXA=sYB$ is composed by piecing together subtableaux.
The rule Match$^{\epsilon+}$ is applied to every potentially successful leaf of a subtableau.
The definition of successful tableaux is the same as in Section~\ref{sec-Decidability of PDAminus}.
Lemma~\ref{2013-08-12} remains valid for $\mathrm{nPDA}^{\epsilon+}$.
Without further ado we state the main result.

\begin{theorem}\label{pushing-nPDA-decidability}
The branching bisimilarity on $\mathrm{nPDA}^{\epsilon+}$ processes is decidable.
\end{theorem}

\section{High Undecidability of $\epsilon$-Nondeterminism}\label{sec-High-Undecidability}

In this section we show that the branching bisimilarity is highly undecidable on $\mathrm{PDA}^{\epsilon+}$.
This is done by a reduction from the $\Sigma_1^{1}$-complete problem \emph{rec-NMCM}.
A \emph{nondeterministic Minsky counter machin} $\mathcal{M}$ with two counters $c_1,c_2$ is a program of the form $1:
I_1;\ 2: I_2;\ \dots;\ n{-}1: I_{n-1};\ n: \textrm{halt}$, where for each $i\in\{1,\dots,n-1\}$
the instruction $I_i$ is in one of the following forms, assuming $1\leq j,k\leq n$ and $e\in\{1,2\}$.
\begin{itemize}
\item $c_{e}:= c_{e}+1$ and then goto $j$.
\item if $c_{e} = 0$ then goto $j$; otherwise $c_{e}:= c_{e}-1$ and then goto $k$.
\item goto $j$ or goto $k$;
\end{itemize}
The problem rec-NMCM asks if $\mathcal{M}$ has an infinite computation on $(c_1,c_2) = (0, 0)$ such that $I_1$ is executed infinitely often.
We shall use the following fact~\cite{Harel1986}.
\begin{proposition}
\emph{rec-NMCM} is $\Sigma_1^{1}$-complete.
\end{proposition}

Following~\cite{JancarSrba2008} we  transform a nondeterministic Minsky counter machine $\mathcal{M}$ with two counters $c_1$ and $c_2$ into a machine $\mathcal{M'}$ with three counters $c_1,c_2,c_3$.
The machine $\mathcal{M}'$ makes use of a new nondeterministic instruction of the form
\[i: c_e :=* \textrm{ and then goto } j.\]
The effect of this instruction is to set $c_e$ by a nondeterministically chosen number and then go to $I_j$.
Every instruction ``$i : I_i$'' of $\mathcal{M}$ is then replaced by two instructions in $\mathcal{M}'$, with respective labels $2i{-}1$ and $2i$.
\begin{itemize}
\item $1: I_1$ is replaced by
  \begin{itemize}
    \item [] $1: c_3 :=*$ and goto $2$;
    \item [] $2: I_1$.
  \end{itemize}
\item $i: I_i$, where $i \in \{2, \dots, n\}$, is replaced by
\begin{itemize}
\item [] $2i-1:$ if $c_3=0$ then goto $2n$; otherwise $c_3 := c_3-1$ and goto $2i$;
\item [] $2i: I_i$
\end{itemize}
\item Inside each $I_i$, where $i \in \{1,\dots, n\}$, every occurrence of ``goto $j$'' is replaced by ``goto $2j-1$''.
\end{itemize}

It is easy to see that $\mathcal{M}'$ has a infinite computation if and only if $\mathcal{M}$ has a infinite computation that uses instruction $1$ infinitely often.
Our goal is to construct a $\mathrm{PDA}^{\epsilon+}$ system $\mathcal{G}=\{\mathcal{Q}, \mathcal{L}, \mathcal{V}, \mathcal{R}\}$ in which we can define two processes $p_1X\bot$ and $q_1X\bot$ that render true the following equivalence.
\[p_1X\bot \beq q_1X\bot \textrm{ iff }\mathcal{M}'\ \textrm{has a infinite computation}.\]
The system $\mathcal{G}=\{\mathcal{Q}, \mathcal{L}, \mathcal{V}, \mathcal{R}\}$ contains the following key elements:
\begin{itemize}
\item Two states $p_i,q_i \in \mathcal{Q}$ are introduced for each instruction $I_i$.
\item $\mathcal{L}=\{a,b,c,c_1, c_2, c_3, f, f'\}$.
\item Three stack symbols $C_1,C_2,C_3\in\mathcal{V}$ are introduced for the three counters respectively.
A bottom symbol $\bot\in\mathcal{V}$ is also introduced.
\end{itemize}
Our construction borrows ideas from~\cite{Mayr2003-ICALP,JancarSrba2008,YinFuHeHuangTao2014}, making use of the game characterization of branching bisimulation and~\emph{Defender's Forcing} technique.
A configuration of $\mathcal{M}'$ that consists of instruction label $i$ and counter values $(c_1, c_2, c_3)=(n_1, n_2, n_3)$ is represented by the game configuration $(p_iXC_1^{n_1}C_2^{n_2}C_3^{n_3}\bot, q_iXC_1^{n_1}C_2^{n_2}C_3^{n_3}\bot)$.
In the rest of the section we shall complete the definition of $\mathcal{G}$ and explain its working mechanism.

\subsection{Test on Counter}

\begin{figure}[t]
\begin{itemize}
\item $tC_1\act{c_1}t$, \ $tC_2\act{c_2}t$, \ $tC_3\act{c_3}t$;

$t'C_1\act{c_1}t'$, \ $t'C_2\act{c_2}t'$,\ $t'C_3\act{b}t\bot$;
\vspace*{1.5mm}

\item $t(e,+)C_j\act{c_j}t(e,+)$ if $j < e$, \ $t(e,+)C_j\act{c_e}tC_j$ if $j\geq e$, \ $t(e,+)\bot\act{c_e}t\bot$;

$t'(e,+)C_j\act{c_j}t$;
\vspace*{1.5mm}
\item $t(e,*)C_1\act{c_1}t(e,*)$,\  $t(e,*)C_2\act{c_2}t(e,*)$,\ $t(e,*)C_3\act{b}t\bot$;

$t'(e,*)C_1\act{c_1}t(e,*)$,\  $t'(e,*)C_2\act{c_2}t(e,*)$,\ $t'(e,*)C_3\act{b}t\bot$;
\vspace{1.5mm}

\item $t(e,-)C_j\act{c_j}t$;

$t'(e,-)C_j\act{c_j}t'(e,-)$ if $j < e$, \ $t'(e,-)C_j\act{c_e}tC_j$ if $j\geq e$, \ $t'(e,-)\bot\act{c_e}t\bot$;
\vspace*{1.5mm}

\item $t(e,0)C_j\act{c_j}t(e,0)$ if $j \neq e$, \ $t(e,0)C_e \act{f}t(e,0)$;

$t'(e,0)C_j\act{c_j}t(e,0)$ if $j \neq e$, \ $t'(e,0)C_e \act{f'}t(e,0)$;
\vspace*{1.5mm}

\item $t(e,1)C_j\act{c_j} t(e,1)$ if $j < e$, \ $t(e,1)C_e\act{c_e}t$, \ $t(e,1)C_j\act{f}t$ if
  $j > e$; \ $t(e,1)\bot\act{f}t\bot$;

$t'(e,1)C_j\act{c_j} t'(e,1)$ if $j < e$, \ $t'(e,1)C_e\act{c_e}t$, \ $t'(e,1)C_j\act{f'}t$ if
$j> e$; \ $t'(e,1)\bot\act{f'}t\bot$;
\vspace*{1.5mm}

\item $p\bot\act{b}t\bot$ for every $p\in\{t,t',t(e,+),t'(e,+),t(e,-),t'(e,-),t(e,0),t'(e,0),t(e,1),t'(e,1)\}$.
\end{itemize}
\caption{Test on Counter}
\label{Testing-on-Counter}
\end{figure}

We need some rules to carry out testing on the counters.
In the rules given in Fig.~\ref{Testing-on-Counter}, $j$ and $e$ range over the set $\{1, 2, 3\}$.
These rules are straightforward.
The following proposition summarizes the correctness requirement on the equality test, the successor and predecessor tests, and the zero test.
Its routine proof is omitted.

\begin{proposition}\label{lm:checking}
Let $\alpha=C_1^{n_1}C_2^{n_2}C_3^{n_3}$ and $\beta=C_1^{m_1}C_2^{m_2}C_3^{m_3}$.
The following statements are valid.
\begin{enumerate}
\item $t\alpha\bot \beq t\beta\bot$ iff $n_e = m_e$ for  $e=1,2,3$.
\item $t(3,*)\alpha\bot \beq t'(3,*)\beta\bot$ iff $n_e = m_e$ for  $e=1,2$.
\item $t(e,+)\alpha\bot \beq t'(e,+)\beta\bot$ iff $n_e+1= m_e$ and $n_j = m_j$ for $j\neq e$.
\item $t(e,-)\alpha\bot \beq t'(e,-)\beta\bot$ iff $n_e = m_e +1$ and $n_j = m_j$ for $j\neq e$.
\item $t(e,0)\alpha\bot \beq t'(e,0)\beta\bot$ iff $n_j = m_j $ for $j = 1,2,3$ and $n_e=0$.
\item $t(e,1)\alpha\bot \beq t'(e,1)\beta\bot$ iff $n_j = m_j $ for $j =1,2,3$ and $n_e > 0$.
\item  $p\alpha\bot \beq p\alpha\bot\beta$ for all $p \in \mathcal{Q}$ and all $\alpha, \beta \in\mathcal{V}^{*}$.
\end{enumerate}
\end{proposition}

\subsection{Operation on Counter}

There are three basic operations on the counters, the increment operation, the decrement operation and the nondeterministic assignment operation.
Our task is to encode these operations in the branching bisimulation game $\mathcal{G}$.
To do that we use a technique from~\cite{YinFuHeHuangTao2014}, which is a refinement of Defender's Forcing technique~\cite{JancarSrba2008}, taking into account of the subtlety of the branching bisimulation.
The idea can be explained using the following system.
\begin{enumerate}
\item $P \act{a} P'$, $P \act{\epsilon} Q_0$. The latter is the only silent transition of $P$.
\item $Q\act{\epsilon} Q_0$. This is the only transition $Q$ may perform. Hence $Q\simeq Q_0$.
\item $Q_0\beq Q'$ whenever $Q_0\Longrightarrow Q'$.
\end{enumerate}
Condition 1 and condition 2 guarantee that $P\simeq Q$ if and only if $P\simeq Q_0$.
So the effectiveness of the Defender's Forcing the copycat rules $P \act{\epsilon} Q_0$, $Q\act{\epsilon} Q_0$ intend to achieve depends on how we define $Q_0$.
Condition 3 is forced upon us by the previous two conditions.
A standard approach to meet the requirement 3 is to make sure that everything that has been done to derive $Q_0\Longrightarrow Q'$ can be undone.
In our setting this is accomplished by starting all over again with the help of the bottom symbol $\bot$.
Once we know that condition 3 is indeed satisfied, the argument for the correctness of the bisimulation game can be simplified in the following sense:
In the game of $(P, Q)$ Attacker would play $P \act{a} P'$.
Defender's optimal response must be of the following form
\[Q \act{\epsilon} Q_0 \act{\epsilon} Q_1 \act{\epsilon}Q_2 \act{\epsilon} \dots \act{\epsilon} Q_k \act{a} Q'.\]
For both players only the configuration $(P',Q')$ need be checked.

\begin{figure}[t]
\begin{itemize}
\item $u(e,o,j)X \act{a} u_1(e,o,j)X$, \ $u(e,o,j)X \act{\epsilon} r'(e,o,j)X$;

$u'(e,o,j)X \act{\epsilon} r'(e,o,j)X$;
\vspace*{1.5mm}

\item $r'(e,o,j)X \act{\epsilon} g'(e,o,j)X\bot$;

$g'(e,o,j)X \act{\epsilon} g'(e,o,j)X_3$;

$g'(e,o,j)X_3 \act{\epsilon} g'(e,o,j)X_3C_3$, \ $g'(e,o,j)X_3 \act{\epsilon} g'(e,o,j)X_2$;

$g'(e,o,j)X_2 \act{\epsilon} g'(e,o,j)X_2C_2$, \ $g'(e,o,j)X_2 \act{\epsilon} g'(e,o,j)X_1$;

$g'(e,o,j)X_1 \act{\epsilon} g'(e,o,j)X_1C_1$, \ $g'(e,o,j)X_1 \act{\epsilon} r'(e,o,j)X$;
\vspace*{1.5mm}

\item $g'(e,o,j)X_1 \act{a} u_1'(e,o,j)X$;
\vspace*{1.5mm}

\item $u_1(e,o,j)\act{a} u_2(e,o,j)X$, \ $u_1(e,o,j)X\act{c} t(e,o)$;

$u_1'(e,o,j)\act{a} u_2'(e,o,j)X$, \ $u_1'(e,o,j)X\act{c} t'(e,o)$;
\vspace*{1.5mm}

\item $u_2(e,o,j)X \act{\epsilon} r(e,o,j)X$;

$u_2'(e,o,j)X\act{\epsilon} r(e,o,j)X$, \ $u_2'(e,o,j)X\act{a}u_3'(e,o,j)X$;
\vspace*{1.5mm}

\item $r(e,o,j)X \act{\epsilon} g(e,o,j)X\bot$; \ $g(e,o,j)X \act{\epsilon} g(e,o,j)X_3$;

$g(e,o,j)X_3 \act{\epsilon} g(e,o,j)X_3C_3$, \ $g(e,o,j)X_3 \act{\epsilon} g(e,o,j)X_2$;

$g(e,o,j)X_2 \act{\epsilon} g(e,o,j)X_2C_2$, \ $g(e,o,j)X_2 \act{\epsilon} g(e,o,j)X_1$;

$g(e,o,j)X_1 \act{\epsilon} g(e,o,j)X_1C_1$, \ $g(e,o,j)X_1 \act{\epsilon} r(e,o,j)X$;
\vspace*{1.5mm}

\item $g(e,o,j)X_1 \act{a} u_3(e,o,j)X$;
\vspace*{1.5mm}

\item $u_3(e,o,j)X\act{a} p_jX$, \ $u_3(e,o,j)X\act{c} t$;

$u_3'(e,o,j)X\act{a} q_jX$, \ $u_3'(e,o,j)X\act{c} t$.
\end{itemize}

\caption{Operation on Counter}
\label{Operation-on-Counter}
\end{figure}

With the above remark in mind we turn to the part of the game that implements the basic operations.
Let $e$ range over $\{1,2,3\}$, $o$ over $\{+,-,*\}$, and $j$ over $\{1,\dots, 2n\}$.
For each triple $(e,o,j)$ we introduce the rules given in Fig.~\ref{Operation-on-Counter}.
The following lemma identifies some useful state preserving silent transitions.

\begin{lemma}\label{lm:forcing}
$P\beq r(e,o,j)X\bot$ for all $P$ such that $r(e,o,j)X\bot \Longrightarrow P$.
Similarly $Q\beq r'(e,o,j)X\bot$ for all $Q$ such that $r'(e,o,j)X\bot \Longrightarrow Q$.
\end{lemma}
\begin{proof}
Suppose $r(e,o,j)X\bot \Longrightarrow P$.
Then $P\Longrightarrow r(e,o,j)X\bot \alpha$ for some $\alpha$.
By (7) of Proposition~\ref{lm:checking} one has $r(e,o,j)X\bot\simeq r(e,o,j)X\bot\alpha$.
Consequently $r(e,o,j)X\bot \simeq P$.
\qed\end{proof}

The next lemma states the soundness property of the rules defined in Fig.~\ref{Operation-on-Counter}, in which we write $\mathbf{1}^{1}$, $\mathbf{1}^{2}$ and $\mathbf{1}^{3}$ respectively for $(1,0,0)$, $(0,1,0)$ and $(0,0,1)$.

\begin{lemma}\label{lm:op_counter}
Suppose $\alpha=C_1^{m_1}C_2^{m_2}C_3^{m_3}$.
The following statements are valid.
 \begin{enumerate}
 \item In the bisimulation of $(u(e,+,j)X\alpha\bot, u'(e,+,j)X\alpha\bot)$ Defender, respectively Attacker, has a strategy to win or at least push the game to $(P,Q)$ such that $P\beq
   p_jXC_1^{n_1}C_2^{n_2}C_3^{n_3}\bot$ and $Q\beq q_jXC_1^{n_1}C_2^{n_2}C_3^{n_3}\bot$ and $(n_1,n_2,n_3) = (m_1, m_2, m_3) {+} \mathbf{1}^{e}$.
\vspace*{1.5mm}

 \item If $m_e > 0$ then in the bisimulation game of $(u(e,-,j)X\alpha\bot,u'(e,-,j)X\alpha\bot)$ Defender, respectively Attacker, has a strategy to win or at least push the
   game to $(P,Q)$ such that $P\beq p_jXC_1^{n_1}C_2^{n_2}C_3^{n_3}\bot$ and $Q\beq q_jXC_1^{n_1}C_2^{n_2}C_3^{n_3}\bot$ and $(n_1,n_2,n_3) = (m_1, m_2, m_3) {-} \mathbf{1}^{e}$.
\vspace*{1.5mm}

 \item Suppose $n \geq 0$.
 In the bisimulation game of $(u(3,*,j)X\alpha\bot,u'(3,*,j)X\alpha\bot)$ Defender has a strategy to win or at least push the game to $(P,Q)$ such that $P \beq p_j XC_1^{n_1}C_2^{n_2}C_3^{n_3}\bot$ and $Q \beq q_jXC_1^{n_1}C_2^{n_2}C_3^{n_3}\bot$ and $(n_1,n_2,n_3) = (m_1, m_2, m_3)+ (n-m_3)\cdot \mathbf{1}^{3}$.
 \end{enumerate}
\end{lemma}
\begin{proof}
We prove the first statement.
The proof for the other two is similar.
Let $\beta=C_1^{n_1}C_2^{n_2}C_3^{n_3}$ such that $(n_1,n_2,n_3) = (m_1, m_2, m_3) {+} \mathbf{1}^{e}$.
In what follows we describe Defender and Attacker's step-by-step optimal strategy in the bisimulation game of $(u(e,+,j)X\alpha\bot, u'(e,+,j)X\alpha\bot)$.
\begin{enumerate}[(i)]
\item By Defender's Forcing, Attacker plays $u(e,+,j)X\alpha\bot \act{a}u_1(e,+,j)X\alpha\bot$.
Defender responds with  \[u'(e,+,j)X\alpha\bot \stackrel{\epsilon}{\Longrightarrow} g'(e,+,j)X_1\beta\bot\alpha\bot  \act{a} u_1'(e,+,j)X\beta\bot\alpha\bot.\]
According to Lemma \ref{lm:forcing} Attacker's optimal move is to continue the game from \[(u_1(e,+,j)X\alpha\bot,u_1'(e,+,j)X\beta\bot\alpha\bot).\]

\item It follows from Proposition~\ref{lm:checking} that $t(e,+)X\alpha\bot \beq t'(e,+)X\beta\bot\alpha\bot$.
If Attacker plays an action labeled $c$, Defender wins.
Attacker's optimal move is to play an action labeled $a$.
Defender then follows suit, and the game reaches the configuration $(u_2(e,+,j)X\alpha\bot,u_2'(e,+,j)X\beta\bot\alpha\bot)$.

\item Attacker's next move is $u_2'(e,+,j)X\beta\bot\alpha\bot \act{a}u_3'(e,+,j)X\beta\bot\alpha\bot$.
This is optimal by Proposition \ref{lm:checking}.
Defender responds with
  \[u_2(e,+,j)X\alpha\bot \stackrel{\epsilon}{\Longrightarrow} g(e,+,j)X_1\beta\bot\alpha\bot {\act{a}} u_3(e,+,j)X\beta\bot\alpha\bot.\]
By an argument similar to the one given in (i) Attacker would choose \[(u_3(e,+,j)X\beta\bot\alpha\bot, u_3'(e,+,j)X\beta\bot\alpha\bot)\]
as the next configuration.

\item If Attacker plays an action labeled $c$, Defender wins by Proposition \ref{lm:checking}.
So Attacker's best bet is to play an action labeled by $a$.
The game reaches the configuration $(p_jX\beta\bot\alpha\bot, q_jX\beta\bot\alpha\bot)$.
\end{enumerate}
The above argument shows that the configuration $(p_jX\beta\bot\alpha\bot, q_jX\beta\bot\alpha\bot)$ is optimal for both Attacker and Defender.
We are done.
\qed\end{proof}

\subsection{Control Flow}

We now encode the control flow of $\mathcal{M}'$ by the rules of the bisimulation game.
We will introduce a number of rules for each instruction in $\mathcal{M}'$.
\begin{enumerate}
\item The following rules are introduced in the game $\mathcal{G}$ for an instruction of the form ``$i:$ $c_e := c_e+1$ and then goto $j$''.
\[ p_iX\act{a} u(e,+,j)X, \quad q_iX\act{a} u'(e,+,j)X.\]
\item For each instruction of the form  ``$i: c_e :=*$ and then goto $j$'' the following two rules are added to $\mathcal{R}$.
\[p_iX\act{a} u(e,*,j)X, \quad q_iX\act{a} u'(e,*,j)X.\]
\item For each instruction of the form ``$i:$ goto $j$ or goto $k$'', we have the following.
\begin{itemize}
\vspace*{1.5mm}
\item $p_iX\act{a}p_i^{1}X$, \ $p_iX\act{a}q_i^{1}X$, \ $p_iX\act{a} q_i^{2}X$;

$q_iX\act{a}q_i^{1}X$, \ $q_iX\act{a} q_i^{2}X$;

\vspace*{1.5mm}

\item $p_i^{1}X \act{a}p_jX$, \ $p_i^{1}X\act{a}p_kX$;

$q_i^{1}X\act{a} q_jX$, \ $q_i^{1}X\act{a} p_kX$;

$q_i^{2}X\act{a} p_jX$, \ $q_i^{2}X\act{a} q_kX$.

\vspace*{1.5mm}

\end{itemize}
These rules embodies precisely the idea of Defender's Forcing \cite{JancarSrba2008}.
It is Defender who makes the choice.
 \item For each instruction of the form
\begin{center}
``$i:$ if $c_e = 0$ then goto $j$; otherwise $c_e = c_e -1$ and then goto $k$''
\end{center}
we construct a system defined by the following rules.
\begin{itemize}
\vspace*{1.5mm}

\item $p_iX\act{a}p_i(e,0,j)$, \ $p_iX\act{c} p_i(e,1,k)$;

$q_iX\act{a}q_i(e,0,j)$, \ $q_iX\act{c}q_i(e,1,k)$;
\vspace*{1.5mm}

\item $p_i(e,0,j)X\act{a}v_1(e,0,j)X$, \ $p_i(e,1,k)X\act{a}v_1(e,1,k)X$;

$p_i(e,0,j)X\act{a}v_2(e,0,j)X$, \ $p_i(e,1,k)X\act{a}v_2(e,1,k)X$;

$p_i(e,0,j)X\act{a}v_3(e,0,j)X$,  $p_i(e,1,k)X\act{a}v_3(e,1,k)X$;
\vspace*{1.5mm}

\item $q_i(e,0,j)X\act{a}v_2(e,0,j)X$, \ $q_i(e,1,k)X\act{a}v_2(e,1,k)X$;

$q_i(e,0,j)X\act{a}v_3(e,0,j)X$, \ $q_i(e,1,k)X\act{a}v_3(e,1,k)X$;
\vspace*{1.5mm}

\item $v_1(e,0,j)X\act{a}t(e,1)X$, \ $v_1(e,0,j)X\act{a}p_jX$;

$v_2(e,0,j)X\act{a}t'(e,1)X$, \ $v_2(e,0,j)X\act{a}p_jX$;

$v_3(e,0,j)X\act{a}t(e,1)X$, \ $v_3(e,0,j)X\act{a}q_jX$;
\vspace*{1.5mm}

\item $v_1(e,1,k)X\act{a}t(e,0)X$, \ $v_1(e,1,k)X\act{a}u(e,-,k)X$;

$v_2(e,1,k)X\act{a}t'(e,0)X$, \ $v_2(e,1,k)X\act{a}u(e,-,k)X$;

$v_3(e,1,k)X\act{a}t(e,0)X$, \ $v_3(e,1,k)X\act{a}u'(e,-,k)X$.
\vspace*{1.5mm}
\end{itemize}
The idea of the above encoding is that Attacker must claim either ``$c_e=0$'' or ``$c_e >0$''.
Defender can check the claim and wins if Attacker lies.
If Attacker has not lied, Defender can force Attacker to  do what Defender wants.

\item For ``$2n: \textrm{halt}$'', we add the rules
\[p_{2n}X \act{f} p_{2n}\bot, \quad q_{2n}X\act{f'} q_{2n}\bot.\]
So Attacker wins if the game ever terminates.
\end{enumerate}
This completes the definition of $\mathcal{G}$.

With the help of Proposition~\ref{lm:checking} and Lemma~\ref{lm:op_counter}, it is a routine to prove the next lemma.
\begin{lemma}\label{2014-04-01}
$\mathcal{M}'$ has a infinite computation if and only if $p_1X\bot\beq q_1X\bot$.
\end{lemma}
Branching bisimilarity on $\textrm{PDA}^{\epsilon+}$ is in $\Sigma_1^{1}$ for the following reason:
For any $\textrm{PDA}^{\epsilon+}$ processes $P$ and $Q$, $P\beq Q$ if and only if there exists a set of pairs that contains $(P,Q)$ and satisfies the first order arithmetic definable conditions prescribed in Definition~\ref{2014-03-30}.
Together with the reduction justified by Lemma~\ref{2014-04-01} we derive the main result of the section.
\begin{theorem}
Branching bisimilarity is $\Sigma_1^{1}$-complete on $\textrm{PDA}^{\epsilon+}$.
\end{theorem}

It has been proved in~\cite{YinFuHeHuangTao2014} that the branching bisimilarity is undecidable on normed PDA.
The reduction defined in the above can be constructed for nPDA too.
This is because in nPDA the stack can be {\em reset} by popping off all the symbols in the stack using $\epsilon$-popping transitions and creating new stack content using $\epsilon$-pushing transitions, achieving the same effect as the bottom symbol $\bot$ achieves in $\textrm{PDA}^{\epsilon+}$.
The details are omitted.

\begin{theorem}
Branching bisimilarity is $\Sigma_1^{1}$-complete on $\textrm{nPDA}$.
\end{theorem}

\section{Conclusion}\label{sec-Conclusion}

\begin{figure*}[t]
\begin{center}
\begin{tabular}{|c|c||c|c|} \hline
 & $\epsilon$-Popping {\bf nPDA}/{\bf PDA} & $\epsilon$-Pushing {\bf nPDA} & $\epsilon$-Pushing {\bf PDA} \\ \hline\hline
$\simeq$ & {\color{blue}Decidable} & {\color{blue}Decidable} & {\color{blue}$\Sigma^{1}_{1}$-Complete} \\ \hline
$\;\approx\;$ & \;$\Pi_{1}^{0}$-Complete~\cite{JancarSrba2008}\; & \;$\Pi_{1}^{0}$-Complete~\cite{JancarSrba2008}\; & \;$\Sigma^{1}_{1}$-Complete~\cite{JancarSrba2008}\; \\ \hline\hline
\end{tabular}
\end{center}
\caption{Dividing Line for PDA \label{Dividing-Line-for-PDA}}
\end{figure*}

The main results of the paper is summarized in Fig.~\ref{Dividing-Line-for-PDA}.
Our decidability results subsume all previous decidability results on the language equivalence on DPDA and the strong bisimilarity on PDA.
The structural definition of PDA plays an important role in the noticeably simpler proofs.
Towards the end of writing up this paper we became aware of the relationship between our definition of PDA and Jan\v{c}ar's notion of first order grammar~\cite{Janvcar2012-FO-Grammar}.
In our opinion Jan\v{c}ar's approach is an abstraction of the issue at a more basic level.
It appears to us that the proof methodology used in this paper can be applied to the first order grammar in a straightforward manner.
We are currently looking into the issue of whether anything new can be said in this abstract setting.

Stirling proved~\cite{Stirling2002} that the language equivalence of DPDA is primitive recursive.
Benedikt, Goller, Kiefer and Murawski showed that the strong bisimilarity on nPDA is non-elementary~\cite{BenediktMollerKieferMurawski2013}.
More recently Jan\v{c}ar observed that the strong bisimilarity of first-order grammar is
Ackermann-hard~\cite{Jancar2014},
a consequence of which is that the strong bisimilarity proved decidable by S{\'e}nizergues in
~\cite{Senizergues1998} is Ackermann-hard.
In view of the stronger results obtained in this paper, it is an interesting research direction to look for tighter upper and lower bounds on the branching bisimilarity of $\mathrm{nPDA}^{\epsilon+}$, $\mathrm{nPDA}^{\epsilon-}$ and $\mathrm{PDA}^{\epsilon-}$.

\vspace*{5mm}\noindent {\bf Acknowledgement}.
We thank the members of BASICS for their interest and constructive questions.
The support from NSFC (61033002, ANR 61261130589, 91318301) is gratefully acknowledged.

\newpage


\begin{thebibliography}{10}
\bibitem{BenediktMollerKieferMurawski2013}
M.~Benedikt, S.~Moller, S.~Kiefer, and A.~Murawski.
\newblock Bisimilarity of pushdown automata is nonelementary.
\newblock In {\em Logic in Computer Science}, pages 488--498, 2013.

\bibitem{BurkartCaucalMollerSteffen2001}
O.~Burkart, D.~Caucal, F.~Moller, and B.~Steffen.
\newblock Verification on infinite structures.
\newblock In J.~Bergstra, A.~Ponse, and S.~Smolka, editors, {\em Handbook of
  Process Algebra}, pages 545--623. North-Holland, 2001.

\bibitem{GinsburgGreibach1966}
S.~Ginsburg and S.~Greibach.
\newblock Deterministic context free languages.
\newblock {\em Information and Control}, 9:620--648, 1966.

\bibitem{GrooteHuttel1994}
J.~Groote and H.~H\"{u}ttel.
\newblock Undecidable equivalences for basic process a;lgebra.
\newblock {\em Information and Computation}, 115:354--371, 1994.

\bibitem{Harel1986}
D.~Harel.
\newblock Effective transformations on infinite trees, with applications to
  high undecidability, dominoes, and fairness.
\newblock {\em J. ACM}, 33:224--248, 1986.

\bibitem{HopcroftUllman1979}
J.~Hopcroft and J.~Ullman.
\newblock {\em Introduction to Automata Theory, Languages and Computation}.
\newblock Addison-Wesley Publishing Company, 1979.

\bibitem{Huttel1992}
H.~H{\"u}ttel.
\newblock Silence is golden: Branching bisimilarity is decidable for context
  free processes.
\newblock In {\em CAV'91}, pages 2--12. Lecture Notes in Computer Science 575,
  Springer, 1992.

\bibitem{Huttel1994}
H.~H{\"u}ttel.
\newblock Undecidable equivalences for basic parallel processes.
\newblock In {\em Theoretical Aspects of Computer Software}, Lecture Notes in
  Computer Science 789, pages 454--464, 1994.

\bibitem{HuttelStirling1991}
H.~H\"{u}ttel and C.~Stirling.
\newblock Actions speak louder than words: Proving bisimilarity for
  context-free processes.
\newblock In {\em LICS'91}, pages 376--386, 1991.

\bibitem{Janvcar2012-FO-Grammar}
P.~Jan\v{c}ar.
\newblock Decidability of dpda language equivalence via first-order grammars.
\newblock In {\em 27th Annual IEEE Symposium on Logic in Computer Science},
  page 415¨C424. IEEE Computer Society, 2012.

\bibitem{Jancar2014}
P.~Jan\v{c}ar.
\newblock Equivalences of pushdown systems are hard.
\newblock {\em Foundations of Software Science and Computation}, pages 1--28,
  2014.

\bibitem{JancarSrba2008}
P.~Jan\v{c}ar and J.~Srba.
\newblock Undecidability of bisimilarity by defender's forcing.
\newblock {\em Journal of ACM}, 55(1), 2008.

\bibitem{Mayr2003-ICALP}
E.~Mayr.
\newblock Undecidability of weak bisimulation equivalence for 1-counter
  processes.
\newblock In {\em ICALP 2003}, Lecture Notes in Computer Science 2719, page
  570¨C583. Springer, 2003.

\bibitem{Mayr2000PRS}
R.~Mayr.
\newblock Process rewrite systems.
\newblock {\em Information and Computation}, 156:264--286, 2000.

\bibitem{Milner1989}
R.~Milner.
\newblock {\em Communication and Concurrency}.
\newblock Prentice Hall, 1989.

\bibitem{Park1981}
D.~Park.
\newblock Concurrency and automata on infinite sequences.
\newblock In {\em Theoretical Computer Science}, volume Lecture Notes in
  Computer Science 104, pages 167--183. Springer, 1981.

\bibitem{Senizergues1997}
G.~S{\'e}nizergues.
\newblock The equivalence problem for deterministic pushdown automata is
  decidable.
\newblock In {\em ICALP 1997}, volume 1256 of {\em Lecture Notes in Computer
  Science}, pages 671--681. Springer-Verlag, 1997.

\bibitem{Senizergues1998}
G.~S{\'e}nizergues.
\newblock Decidability of bisimulation equivalence for equational graphs of
  finite out-degree.
\newblock In {\em Foundations of Computer Science, 1998. Proceedings. 39th
  Annual Symposium on}, pages 120--129. IEEE, 1998.

\bibitem{Senizergues2001}
G.~S{\'e}nizergues.
\newblock L (a)= l (b)? decidability results from complete formal systems.
\newblock {\em Theoretical Computer Science}, 251(1-2):1--166, 2001.

\bibitem{Senizergues2002}
G.~S{\'e}nizergues.
\newblock L (a)= l (b)? a simplified decidability proof.
\newblock {\em Theoretical Computer Science}, 281(1):555--608, 2002.

\bibitem{Srba2002d}
J.~Srba.
\newblock Undecidability of weak bisimilarity for pushdown processes.
\newblock In {\em ONCUR 2002}, volume 2421 of {\em LNCS}, pages 579--593.
  Springer-Verlag, 2002.

\bibitem{Stirling2002}
Stirling.
\newblock Deciding dpda equivalence is primitive recursive.
\newblock In {\em ICALP 2002}, Lecture Notes in Computer Science 2380, pages
  821--832. Springer, 2002.

\bibitem{Stirling1996-nPDA-decidability}
C.~Stirling.
\newblock Decidability of bisimulation equivalence for normed pushdown
  processes.
\newblock In {\em Proceedings of the 13th International Conference on
  Concurrency Theory (CONCUR 1996)}, Lecture Notes in Computer Science, pages
  217--232. Springer-Verlag, 1996.

\bibitem{Stirling1998-nPDA-decidability}
C.~Stirling.
\newblock Decidability of bisimulation equivalence for normed pushdown
  processes.
\newblock {\em Theoretical Computer Science}, 195(2):113--131, 1998.

\bibitem{Stirling1998}
C.~Stirling.
\newblock The joy of bisimulation.
\newblock In {\em MFCS 1998}, Lecture Notes in Computer Science 1450, pages
  142--151. Springer, 1998.

\bibitem{Stirling2000-PDA-decidability}
C.~Stirling.
\newblock Decidability of bisimulation equivalence for pushdown processes.
\newblock 2000.

\bibitem{Stirling2001-DPDA-dcidability}
C.~Stirling.
\newblock Decidability of dpda equivalence.
\newblock {\em Theoretical Computer Science}, 255(1-2):1--31, 2001.

\bibitem{vanGlabbeekWeijland1989-first-paper-bb}
R.~van Glabbeek and W.~Weijland.
\newblock Branching time and abstraction in bisimulation semantics.
\newblock In {\em Information Processing'89}, pages 613--618. North-Holland,
  1989.

\bibitem{vanGlabbeekWeijland1996}
R.~van Glabbeek and W.~Weijland.
\newblock Branching time and abstraction in bisimulation semantics.
\newblock {\em Journal of ACM}, 3:555--600, 1996.

\bibitem{YinFuHeHuangTao2014}
Q.~Yin, Y.~Fu, C.~He, M.~Huang, and X.~Tao.
\newblock Branching bisimilarity checking for prs.
\newblock In {\em ICALP 2014}, to appear.
\end{thebibliography}
\end{document}